\documentclass[a4paper,11pt,openright]{article} % Per avere margini destro e sinistro diversi
\usepackage[english]{babel} % per lingua. TeXnic dà warning per misteriosi motivi...
\usepackage[utf8]{inputenc} 
\usepackage{amsmath, amssymb, stmaryrd,amsthm,amsfonts, amscd,color,enumerate,eucal,latexsym,mathrsfs,mathtools,cancel,tikz}
\usepackage{epsfig, hieroglf, protosem}  % ciao
\usepackage{simplewick, bm}
\usepackage{verbatim}
\usepackage{hyperref}
\usepackage[all,cmtip]{xy}
\usetikzlibrary{matrix,calc}

\numberwithin{equation}{section}

\definecolor{LorColor}{RGB}{230, 74, 129}
\definecolor{PaColor}{RGB}{255,0,0}
\definecolor{NiColor}{RGB}{77,77,255}
\definecolor{ClaColor}{RGB}{0,0,0}
\definecolor{NiColoRed}{RGB}{255,77,77}
\definecolor{NiCitation}{RGB}{0,181,26}
\definecolor{AlColor}{RGB}{77,0,154}

\newcommand{\clacomment}[1]{\begingroup\color{ClaColor}#1\endgroup}

\newtheoremstyle{TheoremStyle}% <name>
{3pt}% <Space above>
{3pt}% <Space below>
{\slshape}% <Body font>
{}% <Indent amount>
{\bf}%{\itshape}% <Theorem head font>
{:}% <Punctuation after theorem head>
{.5em}% <Space after theorem head>
{}% <Theorem head spec (can be left empty, meaning 'normal')>

\theoremstyle{TheoremStyle}
\newtheorem{theorem}{Theorem}[section]
\newtheorem{corollary}[theorem]{Corollary}
\newtheorem{proposition}[theorem]{Proposition}

\newtheorem{definition}[theorem]{Definition}
\newtheorem{remark}[theorem]{Remark}
\newtheorem{example}[theorem]{Example} %[theorem]{Example}

\newcommand{\fish}{\begin{tikzpicture}[thick,scale=1.2]
		\draw[red] (0,0) edge [out=30,in=-30] node[above] {} (0,.35);
		\draw (0,0) edge [out=150,in=210] node[above] {} (0,.35);
\end{tikzpicture}}
\newcommand{\fiammifero}{\begin{tikzpicture}[thick,scale=1.2]
		\draw (0,0) -- (0,0.3);
		\filldraw (0,0.3)circle (1pt);
\end{tikzpicture}}
\newcommand{\propagatore}{\begin{tikzpicture}[thick,scale=1.2]
		\draw (0,0) -- (0,0.3);
\end{tikzpicture}}

\newcommand{\fiammiferoCC}{\begin{tikzpicture}[thick,scale=1.2]
		\draw[red] (0,0) -- (0,0.3);
		\filldraw[red] (0,0.3)circle (1pt);
\end{tikzpicture}}

\newcommand{\propagatoreCC}{\begin{tikzpicture}[thick,scale=1.2]
		\draw[red] (0,0) -- (0,0.3);
\end{tikzpicture}}

\title{An Algebraic and Microlocal Approach to the Stochastic Non-linear Schr\"odinger Equation}

\author{
 Alberto Bonicelli\thanks{AB:
			Dipartimento di Fisica,
		Universit\`a degli Studi di Pavia \& INFN, Sezione di Pavia, 
		Via Bassi 6,
		I-27100 Pavia,
		Italia;
		\mbox{alberto.bonicelli01@universitadipavia.it}
	}
	\and
	Claudio Dappiaggi\footnote{corresponding author}\;\;\thanks{CD:
		Dipartimento di Fisica,
		Universit\`a degli Studi di Pavia \& INFN, Sezione di Pavia, 
		Via Bassi 6,
		I-27100 Pavia,
		Italia;
		\mbox{claudio.dappiaggi@unipv.it}}
\and
Paolo Rinaldi \thanks{PR: Institute for Applied Mathematics, Universit\"at Bonn, 
	Endenicher Allee 60,
	D-53115 Bonn,
	Germany;
	\mbox{rinaldi@iam.uni-bonn.de}
}}

\date{\today}

\begin{document}
\maketitle
\begin{abstract}
	In a recent work \cite{Dappiaggi:2020gge}, it has been developed a novel framework aimed at studying at a perturbative level a large class of non-linear, scalar, real, stochastic PDEs and inspired by the algebraic approach to quantum field theory. The main advantage is the possibility of computing the expectation value and the correlation functions of the underlying solutions accounting for renormalization intrinsically and without resorting to any specific regularization scheme. In this work we prove that it is possible to extend the range of applicability of this framework to cover also the stochastic non-linear Schr\"odinger equation in which randomness is codified by an additive, Gaussian, complex white noise.
\end{abstract}

\paragraph*{Keywords:}
\small{Stochastic Non-linear Schr\"odinger Equation, Algebraic Quantum Field Theory, Renormalization}
\paragraph*{MSC 2020:} 81T05, 60H17.

\tableofcontents

\section{Introduction}\label{Sec: introduction}

Stochastic partial differential equations (SPDEs) play a prominent r\^{o}le in modern analysis, probability theory and mathematical physics due to their effectiveness in modeling different phenomena, ranging from turbulence to interface dynamics. At a structural level, several progresses in the study of non linear problems have been made in the past few years, thanks to Hairer's work on the theory of regularity structures \cite{Hairer14, Hairer15} and to the theory of paracontrolled distributions \cite{Gubinelli} based on Bony paradifferential calculus \cite{Bony}.

Without entering into the details of these approaches, we stress that a common hurdle in all of them is the necessity of applying a renormalization scheme to cope with ill-defined product of distributions. In all these instances the common approach consists of introducing a suitable $\epsilon$-regularization scheme which makes manifest the pathological divergences in the limit $\epsilon\to 0^+$, see {\it e.g.} \cite{Hairer14}. This strategy is very much inspired by the standard approach to a similar class of problems which appears in theoretical physics and, more precisely, in quantum field theory on Minkowski spacetime, studied in momentum space. 

Despite relying on a specific renormalization scheme, all these approaches have been tremendously effective in developing the solution theory of non linear stochastic partial differential equations in presence of an additive white noise. At the same time not much attention has been devoted to computing explicitly the expectation value and the correlations of the underlying solutions, features which are of paramount relevance in the applications, especially when inspired by physics.

In view of these comments, in a recent paper, \cite{Dappiaggi:2020gge}, it has been developed a novel framework to analyze scalar, non linear SPDEs in presence of an additive white noise. The inspiration as well as the starting point for such work comes from the algebraic approach to quantum field theory, see \cite{BFDY15, Rejzner:2016hdj} for reviews. In a few words this is a specific setup which separates on one side observables, collecting them in a suitable unital $*$-algebra encoding specific structural properties ranging from dynamics, to causality and the canonical commutation relations. On the other side, one finds states, that is normalized, positive linear functionals on the underlying algebra, which allow to recover via the GNS theorem the standard probabilistic interpretation of quantum theories. 

Without entering into more details, far from the scope of this work, we stress that this approach, developed to be effective both in coordinate and in momentum space, has the key advantage of allowing an analysis of interacting theories within the realm of perturbation theory \cite{Brunetti:2009qc}. 
Most notably, renormalization plays an ubiquitous r\^{o}le and, following an approach \`a la Epstein-Glaser, this is codified intrinsically within this framework without resorting to any specific \textit{ad hoc} regularization scheme \cite{Brunetti-Fredenhagen-00}.

From a technical viewpoint the main ingredients in this successful approach are a combination of the algebraic structures at the heart of the perturbative series together with the microlocal properties both of the propagators ruling the linear part of the underlying equations of motion and of the two-point correlation function of the chosen state. 

In \cite{Dappiaggi:2020gge} it has been observed that the algebraic approach could as well be adapted also to analyze stochastic, scalar semi-linear SPDEs such as
\begin{align}\label{Eq: Generic equation}
E\widehat{\Phi}=\xi+F[\widehat\Phi].
\end{align}
Here $\widehat{\Phi}$ must be interpreted as a random distribution on the underlying manifold $M$, $\xi$ denotes the standard  Gaussian, real white noise centered at $0$ whose covariance is $\mathbb{E}[\xi(x)\xi(y)]=\delta(x-y)$. 
Furthermore $F:\mathbb{R}\to\mathbb{R}$ is a non-linear potential which can be considered for simplicity of polynomial type, while $E$ is a linear operator either of elliptic or of parabolic type -- see \cite{Dappiaggi:2020gge} for more details and comments.

Following \cite{Brunetti:2009qc} and inspired by the so-called functional formalism \cite{CDDR20, DDR20, Fredenhagen:2014lda}, we consider a  specific class of distributions with values in polynomial functionals over $C^\infty(M)$. 
While we refer a reader interested in more details to \cite{Dappiaggi:2020gge}, we sketch briefly the key aspects of this approach. 
The main ingredients are two distinguished elements
\begin{align*}
	\Phi(f;\varphi)=\int_M \varphi(x) f(x)\mu(x),\quad\boldsymbol{1}(f;\varphi)=\int_M f(x)\mu(x)
\end{align*}
where $\mu$ is a strictly positive density over $M$, $\varphi\in C^\infty(M)$ while $f\in C^\infty_0(M)$. These two functionals are employed as generators of a commutative algebra $\mathcal{A}$ whose composition is the pointwise product. The main rationale at the heart of \cite{Dappiaggi:2020gge} and inspired by the algebraic approach \cite{Brunetti:2009qc} is the following: The stochastic behavior codified by the white noise $\xi$ can be encoded in $\mathcal{A}$ by deforming its product setting for all  $\tau_1,\tau_2\in\mathcal{A}$
\begin{align}\label{Eq: formal delta-product}
	(\tau_1\cdot_Q \tau_2)(f;\varphi)
	&=\sum_{k\geq 0}\frac{1}{k!}[(\delta_{\mathrm{Diag}_2}\otimes Q^{\otimes k})\cdot(\tau_1^{(k)}\otimes \tau_2^{(k)})](f\otimes1_{1+2k};\varphi)
	=\sum_{k\geq 0}t_k(f\otimes1_{1+2k};\varphi)\,,
\end{align}
where $\tau_i^{(k)}$, $i=1,2$ indicate the $k$-th functional derivatives, while $Q=G\circ G^*$. Here $G$ ({\em resp.} $G^*$) is a fundamental solution associated to $E$ ({\em resp.} $E^*$, the formal adjoint of $E$), while $\circ$ indicates the composition of distributions. By a careful analysis of the singular structure of $G$ and in turn of $Q$ one can infer that the distributions $t_k$ are well-defined on $M^{2k+2}$ \clacomment{up to the total diagonal $\mathrm{Diag}_{2k+2}$ of $M^{2k+2}$, that is, $t_k(\cdot\,;\varphi)\in\mathcal{D}'(M^{2k+2}\setminus\mathrm{Diag}_{2k+2})$}.

Yet, adapting to the case in hand the results of \cite{Brunetti-Fredenhagen-00}, in \cite{Dappiaggi:2020gge} it has been proven that it is possible to extend $t_k$ to $\hat{t}_k\in\mathcal{D}^\prime(M^{2k+2})$. This renormalization procedure, when existent, might not be unique, but the ambiguities have been classified, giving ultimately a mathematical precise meaning to Equation \eqref{Eq: formal delta-product}. 
For a further applications of these techniques to the analysis of \textit{a priori} ill-defined product of distributions see \cite{DRS21}.

As a consequence one constructs a deformed algebra $\mathcal{A}_{\cdot Q}$ whose elements encompasses at the algebraic level the information brought by the white noise.
Without entering into many details, which are left to \cite{Dappiaggi:2020gge} for an interested reader, we limit ourselves at focusing our attention on a specific, yet instructive example. More precisely we highlight that, if one evaluates at the configuration $\varphi=0$ the product of two generators $\Phi$ of $\mathcal{A}_{\cdot Q}$, one obtains
$$\left(\Phi\cdot_Q\Phi\right)(f;0)=\widehat{Q\delta_{\mathrm{Diag}_2}}(f),$$
where $f\in C^\infty_0(M)$ and where $=\widehat{Q\delta_{\mathrm{Diag}_2}}$ indicates a renormalized version of the otherwise ill-defined composition between the operator $Q$  and $\delta_{\mathrm{Diag}_2}$. The latter is the standard bi-distribution lying in $\mathcal{D}^\prime(M\times M)$ such that $\delta_{Diag_2}(h)=\int_M\mu(x)\,h(x,x)$ for every $h\in C^\infty_0(M\times M)$.

% If we introduce the the stochastic convolution $\widehat{\varphi}=G\circledast\xi$, whose action on $f\in\mathcal{D}(\mathbb{R}^{d+1}),\varphi\in\mathcal{E}(\mathbb{R}^{d+1})$ is $(G\circledast u)(f;\eta,\eta^\prime)\doteq u(G\circledast f;\eta,\eta^\prime),$

A direct inspection shows that we have obtained the expectation value
$\mathbb{E}(\widehat{\varphi}^2(x))$ of the random field $\widehat{\varphi}^2(x)$, where $\widehat{\varphi}\vcentcolon=G\ast\xi$ is the so-called \emph{stochastic convolution}, which is a solution of Equation \ref{Eq: Generic equation} when $F=0$ and for vanishing initial conditions.  With a similar procedure one can realize that $\mathcal{A}_{\cdot Q}$ encompasses the renormalized expectation values of all finite products of the underlying random field. 

Without entering into further unnecessary details, we stress that an additional extension of the deformation procedure outlined above allows also to identify another algebra whose elements encompass at the algebraic level the information on the correlations between the underlying random fields. All these data have been applied in \cite{Dappiaggi:2020gge} to the $\Phi^3_d$ model. This has been considered as a prototypical case of a nonlinear SPDE and it has been analyzed at a perturbative level constructing both the solutions and their two-point correlations. Most notably renormalization and its associated freedoms have been intrinsically encoded yielding order by order in perturbation theory a renormalized equations which takes them automatically into account. It is worth stressing that, although an analysis of the convergence of the perturbative series is still not within our grasp, a notable advantage of the approach introduced in \cite{Dappiaggi:2020gge} lies in its applicability to a vast class of interactions, including all polynomial ones. On the contrary the existence and uniqueness results for an SPDE based either on the theory of regularity structures or on that of paracontrolled distributions require that the underlying model lies in the so-called subcritical regime. Without entering in the technical details, it entails that it is finite the number of distributions which need to be renormalized. This requirement does not apply to the approach used in this paper, hence allowing an investigation of specific models which cannot otherwise be discussed.

In view of the novelty of \cite{Dappiaggi:2020gge}, several questions are left open and goal of this paper is to address and solve a specific one. Most notably the class of SPDEs considered in the above reference contains only scalar distributions and consistently a real Gaussian white noise. Yet, there exists several models not falling within this class and a notable one, at the heart of this paper, goes under the name of stochastic nonlinear Schr\"odinger equation. 

The reasons to consider this model are manifold. From a physical viewpoint, such specific equation is used in the analysis of several relevant physical phenomena ranging from Bose-Einstein condensates to type II superconductivity when coupled to an external magnetic field \cite{De Dominicis, Stoof, Sasik}. From a mathematical perspective, such class of equations is particularly relevant for its distinguished structural properties and it has been studied by several authors, see {\it e.g.} \cite{Bourgain, Hoshino-pre, Hoshino} for a list of those references which have been of inspiration to this work.

It is worth stressing that in many instances the attention has been given to the existence and uniqueness of the solutions rather than in their explicit construction or in the characterization of the mean and of the correlation functions. It is therefore natural to wonder whether the algebraic approach introduced in \cite{Dappiaggi:2020gge} can be adapted also to this scenario. A close investigation of the system in hand, see Section \ref{Sec: SNLS} for the relevant definitions, unveils that this is not a straightforward transition and a close investigation is necessary. The main reason can be ascribed to the presence of an additive {\em complex} Gaussian white noise, see Equations \eqref{Eq: complex white noise} and \eqref{Eq: complex white noise - 2} which entails that the only non vanishing correlations are those between the white noise and its complex conjugate.
This property has severe consequences, most notably the necessity of modifying significantly the algebraic structure at the heart of \cite{Dappiaggi:2020gge}. 

\clacomment{A further remarkable difference with respect to the cases considered in \cite{Dappiaggi:2020gge} has to be ascribed to the singular structure of the fundamental solution of the Schr\"odinger operator. Without entering here into the technical details, we observe that especially the Epstein-Glaser renormalization procedure bears the consequences of this feature, since this calls for a thorough study of the scaling degree of all relevant integral kernels with respect to the submanifold  $\Lambda^{2k+2}_t=\{(\widehat{t}_{2k+2},\widehat{x}_{2k+2})\,\vert\,t_1=\ldots=t_{2k+2}\}\subset M^{2k+2}$ rather than with respect to the total diagonal $\mathrm{Diag}_{2k+2}$ of $M^{2k+2}$, as in \cite{Dappiaggi:2020gge}.

As we shall discuss in Section \ref{Sec: graph}, the distinguished r\^{o}le of this richer singular structure emerges in the analysis of the subcritical regime which occurs only if the space dimension is $d=1$, contrary to what occurs in many parabolic models as highlighted in \cite{Dappiaggi:2020gge}}. 

The goal of this work will be to reformulate the algebraic approach to SPDEs for the case on a nonlinear Schr\"odinger equation and for simplicity of the exposition we shall focus our attention only to the case of $\mathbb{R}^{d+1}$ as an underlying manifold. Although the extension to more general curved backgrounds \cite{Dappiaggi:2020gge, RS21} is possible with a few minor modifications, we feel that this might lead us astray from the main goal of showing the versatility of the algebraic approach and thus we consider only the scenario which is more of interest in the concrete applications. 

In the next subsection we discuss in detail the specific model that we study in this paper, but, prior to that, we conclude the introduction with a short synopsis of this work. In Chapter \ref{Sec: Preliminaries} we introduce the main algebraic and analytic ingredients necessary in this paper. In particular we discuss the notion of functional-valued distributions, adapting it to the case of an underlying complex valued partial differential equation. In addition we show how to construct a commutative algebra of functionals using the pointwise product and we discuss the microlocal properties of its elements. We stress that we refer to \cite{Hormander-I-03} for the basic notions concerning the wavefront set and the connected operations between distributions, while we rely on \cite[App. B]{Dappiaggi:2020gge} for a concise summary of the scaling degree of a distribution and of its main properties. In Section \ref{Sec: expectations} we prove the existence of $\mathcal{A}^{\mathbb{C}}_{\cdot Q}$ a deformation of the algebra identified in the previous analysis and we highlight how it can codify the information of a complex white noise. Renormalization is a necessary tool in this construction and it plays a distinguished r\^{o}le in Theorem \ref{Thm: deformation map cdot} which is one of the main results of this work. In Section \ref{Sec: correlations} we extend the analysis first identifying a non-local algebra constructed out of $\mathcal{A}^{\mathbb{C}}_{\cdot Q}$ and then deforming its product so to be able to compute multi-local correlation functions of the underlying random field. At last, in Section \ref{Sec: perturbative analysis} we focus our attention on the stochastic nonlinear Schr\"odinger equation. First we show how to construct at a perturbative level the solutions using the functional formalism. In particular we show that, at each order in perturbation theory, the expectation value of the solution vanishes and we compute up to first order the two-point correlation function. As last step, we employ a diagrammatic argument to discuss under which constraint on the dimension of the underlying spacetime, the perturbative analysis to all orders has to cope with a finite number of divergences, to be tamed by means of renormalization. In this respect we greatly improve a similar procedure proposed in \cite[Sec. 6.3]{Dappiaggi:2020gge}. 

\subsection{The Stochastic Nonlinear Schr\"odinger Equation}\label{Sec: SNLS}

In this short subsection, we introduce the main object of our investigation, namely the {\em stochastic nonlinear Schr\"odinger equation} on $\mathbb{R}^{d+1}$, $d\geq 1$.  In addition, throughout this paper, we assume that the reader is familiar with the basic concepts of microlocal analysis, see {\it e.g.} \cite[Ch. 8]{Hormander-I-03}, as well as with the notion of scaling degree, see in particular \cite{Brunetti-Fredenhagen-00} and \cite[App. B]{Dappiaggi:2020gge}.

More precisely, with a slight abuse of notation, by $\mathcal{D}^\prime(\mathbb{R}^d)$ we denote the collection of all random distributions and, inspired by \cite{Hoshino} we are interested in $\psi\in\mathcal{D}^\prime(\mathbb{R}^d)$ such that
\begin{equation}\label{Eq: Stochastic NLS}
	i\partial_t\psi = -\Delta\psi +\lambda|\psi|^{2\kappa}\psi+\xi,
\end{equation}
where $\lambda\in\mathbb{R}$, $\Delta$ is the Laplace operator on the Euclidean space $\mathbb{R}^d$, while $t$ plays the r\^{o}le of the time coordinate along $\mathbb{R}$. The parameter $\kappa\in\mathbb{N}$ controls the nonlinear behaviour of the equation and it is here left arbitrary since, in our approach, it is not necessary to fix a specific value, although a reader should bear in mind that, in almost all concrete models, {\it e.g.} Bose-Einstein condensates or the Ginzburg-Landau theory of type II superconductors, $\kappa=1$. At the same time we remark that the framework that we develop allows to consider a more general class of potentials in Equation \eqref{Eq: Stochastic NLS}, but we refrain from moving in this direction so to keep a closer contact with models concretely used in physical applications. For this reason, unless stated otherwise, henceforth we set $\kappa=1$ in Equation \eqref{Eq: Stochastic NLS}. The stochastic character of this equation is codified in $\xi$, which is a complex, additive Gaussian random distribution, fully characterized by its mean and covariance:
\begin{gather}\label{Eq: complex white noise}
	\mathbb{E}[\xi(f)]=\mathbb{E}[\overline{\xi}(f)]=0,\\
	\mathbb{E}(\xi(f)\xi(h))=\mathbb{E}(\overline{\xi}(f)\overline{\xi}(h))=0,\quad\mathbb{E}(\overline{\xi}(f)\xi(h))=(\overline{f},h)_{L^2},\label{Eq: complex white noise - 2}
\end{gather}
where $f,h\in\mathcal{D}(\mathbb{R}^{d+1})$, while $(,)_{L^2}$ represents the standard inner product in $L^2(\mathbb{R}^{d+1})$. The symbol $\mathbb{E}$ stands for the expectation value.

\begin{remark}
	We stress that we have chosen to work with a Gaussian white noise centered at $0$ only for convenience and without loss of generality. If necessary and mutatis mutandis we can consider a shifted white noise, i.e. $\xi$ is such that the covariance is left unchanged from Equation \eqref{Eq: complex white noise} while
	$$\mathbb{E}[\xi(f)]=\varphi(f)\quad\textrm{and}\quad\mathbb{E}[\overline{\xi}(f)]=\overline{\varphi}(f),$$
	where $\varphi\in\mathcal{E}(\mathbb{R}^{d+1})$, while $\varphi(f)\doteq\int\limits_{\mathbb{R}^{d+1}}dx\,\varphi(x)f(x)$.
\end{remark}

To conclude the section, we observe that the linear part of Equation \eqref{Eq: Stochastic NLS} is ruled by the Schr\"odinger operator $L\doteq i\partial_t+\Delta$, which is a formally self-adjoint operator. For later convenience we introduce the fundamental solution $G\in\mathcal{D}^\prime(\mathbb{R}^{d+1}\times\mathbb{R}^{d+1})$ whose integral kernel reads
\begin{equation}\label{Eq: Integral Kernel}
	G(x,y)=\frac{\Theta(t-t^\prime)}{(4\pi i (t-t^\prime))^{\frac{d}{2}}}e^{-\frac{|\underline{x}-\underline{y}|^2}{4i(t-t^\prime)}},
\end{equation}
where $x=(t,\underline{x})$, $y=(t^\prime,\underline{y})$ while $\Theta$ is the Heaviside function.
\begin{remark}\label{Rem: la vita e' dura}
	Observe that, in comparison to the cases considered in \cite{Dappiaggi:2020gge}, $L$ is not a microhypoelliptic operator. We can estimate the singular structure of $G$ using the following standard microlocal techniques, see \cite{Hormander-I-03}. Since $G$ is a fundamental solution of the Schr\"odinger operator $L=i\partial_t+\Delta$, it holds that
	$$\mathrm{WF}(\delta_{{Diag}_2})\subseteq\mathrm{WF}(G)\subseteq\mathrm{WF}(\delta_{{Diag}_2})\cup\left(\mathrm{Char}(L\otimes\mathbb{I})\cap\mathrm{Char}(\mathbb{I}\otimes L)\right),$$
	where 
	\begin{equation}\label{Eq: Delta Diag_2}
		\delta_{Diag_2}\in\mathcal{D}^\prime(\mathbb{R}^{2(d+1)})\quad|\quad \mathcal{D}(\mathbb{R}^{2(d+1)})\ni f(x_1,x_2)\mapsto \delta_{Diag_2}(f)\doteq\int_{\mathbb{R}^{d+1}}dx_1 f(x_1,x_1)\,,
	\end{equation}
and where $\mathrm{Char}$ denotes the characteristic set of the operator. If we combine this information with the smoothness of Equation \eqref{Eq: Integral Kernel} when $t\neq t^\prime$, we can conclude that
\begin{equation}\label{Eq: form of the WF}
\mathrm{WF}(G)\subseteq\{(t,\underline{x},t,\underline{x},\omega,k,-\omega,-k)\in T^*(\mathbb{R}^{2(d+1)}\setminus\{0\})\}\cup\{(t,\underline{x},t,\underline{y},\omega,0,-\omega,0)\in T^*(\mathbb{R}^{2(d+1)}\setminus\{0\})\}.
\end{equation}
Observe that the first set in Equation \eqref{Eq: form of the WF} is the wavefront set of $\delta_{{Diag}_2}$, while the second one accounts for the contribution of the characteristic set of $L$. It is noteworthy that, since $G(x,y)$ in Equation \eqref{Eq: Integral Kernel} is manifestly singular as $x\neq y$ and $t=t^\prime$, this second contribution to Equation \eqref{Eq: form of the WF} cannot be empty. 

This entails that, being the wavefront set a conical subset, there exist only three possible outcomes: {\em 1)} $\omega>0$, {\em 2)} $\omega<0$ or {\em 3)} $  \omega\in\mathbb{R}\setminus\{0\}$. Alas, to the best of our knowledge, an exact evaluation of the wave front set of $G$ is not present in the literature and it is highly elusive to a direct calculation. From our viewpoint the reason lies in the interplay between the Heaviside function and the oscillatory kernel $$K(x,y)=\frac{e^{-\frac{|\underline{x}-\underline{y}|^2}{4i(t-t^\prime)}}}{(4\pi i (t-t^\prime))^{\frac{d}{2}}}.$$
As a matter of fact, since $\widehat{K}(\omega,k)=\delta(\omega-|k|^2)$ and since $(L\otimes\mathbb{I})K=(\mathbb{I}\otimes L)K=0$, using \cite[Thm. IX.44]{Reed-Simon75}, it descends that 
$$\mathrm{WF}(K)=\{(t,\underline{x},t,\underline{y},\omega,0,-\omega,0\in T^*(\mathbb{R}^{2d+1})\setminus\{0\},\;|\;\omega >0)\}.$$
Here we have used the convention that, for all $\phi\in L^1(\mathbb{R}^d)$, $\widehat{\phi}(k)=\int\limits_{\mathbb{R}^d}dx\,e^{ik\cdot x}\phi(x)$. Hence one cannot use H\"ormander criterion for multiplication of distributions to define $G$ as a product between the Heaviside distribution and $K$. This is the main reason which makes the evaluation of the wavefront set of $G$ rather elusive.

In view of the estimate in Equation \eqref{Eq: form of the WF}, we will consider the worst case scenario, namely the third option above, where $\omega\in\mathbb{R}\setminus\{0\}$. For the sake of the analytical constructions required in this paper, this assumption is not restrictive because we are interested in working with products of distributions of the form $G\cdot G^*$, where $G^*$ is the fundamental solution of the formal adjoint of $L$, namely $i\partial_t+\Delta$. The definition of such product requires renormalization and this distinguished feature is independent from the form of the second component of the right hand side of Equation \eqref{Eq: form of the WF}.
\end{remark}

\begin{remark}[{\bf Notation}]\label{Rem: Cut-Off on G}
	In our analysis, we shall be forced to introduce a cut-off to avoid infrared divergences, namely we consider a real valued function $\chi\in\mathcal{D}(\mathbb{R}^{d+1})$ and we define 
	$$G_\chi\doteq G\cdot(1\otimes \chi).$$
	We stress that this modification does not change the microlocal structure of the underlying distribution, {\it i.e.} $\mathrm{WF}(G_\chi)=\mathrm{WF}(G)$. Therefore, since the cut-off is ubiquitous in our work and in order to avoid the continuous use of the subscript $\chi$, with a slight abuse of notation we will employ only the symbol $G$, leaving $\chi$ understood.
\end{remark}

\section{Analytic and Algebraic Preliminaries}\label{Sec: Preliminaries}

In this section we introduce the key analytic and algebraic tools which are used in this work, also fixing notation and conventions. As mentioned in the introduction, our main goal is to extend and to adapt the algebraic and microlocal approach to a perturbative analysis of stochastic partial differential equations (SPDEs), so to be able to discuss specific models encoding complex random distributions. In particular our main target is the stochastic nonlinear Schr\"odinger equation as in Equation \eqref{Eq: Stochastic NLS} and, for this reason we shall adapt all the following definitions and structures to this case, commenting when necessary on the extension to more general scenarios.

\noindent The whole algebraic and microlocal programme is based on the concept of functional-valued distribution which is here spelt out. We recall that with $\mathcal{E}(\mathbb{R}^d)$ ({\em resp.} $\mathcal{D}(\mathbb{R}^d)$) we indicate the space of smooth ({\em resp.} smooth and compactly supported) complex valued functions on $\mathbb{R}^d$ endowed with their standard locally convex topology.

\begin{remark}\label{Rem: Key difference}
	As mentioned in the introduction, due to the presence of a complex white noise in Equation \eqref{Eq: Stochastic NLS}, we cannot apply slavishly the framework devised in \cite{Dappiaggi:2020gge} since it does not allow to account for the defining properties listed in Equation \eqref{Eq: complex white noise}. Heuristically, one way to bypass this hurdle consists of adopting a different viewpoint which is inspired by the analysis of complex valued quantum fields such as the Dirac spinors, see {\it e.g.} \cite{Araki:1971id}. More precisely, using the notation and nomenclature of Section \ref{Sec: SNLS}, we shall consider $\psi$ and $\bar{\psi}$ as two a priori independent fields, imposing only at the end the constraint that they are related by complex conjugation. As it will become clear in the following, this shift of perspective has the net advantage of allowing a simpler construction of the algebra deformations, see in particular Theorem \ref{Thm: deformation map cdot} avoiding any potential ordering problem between the field and its complex conjugate. This might arise if we do not keep them as independent.
\end{remark}

\begin{definition}\label{Def: Functionals}
	Let $d,m\in\mathbb{N}$. We call {\bf functional-valued distribution}  $u\in\mathcal{D}^\prime(\mathbb{R}^{d+1};\mathsf{Fun}_{\mathbb{C}})$ a map 
	\begin{equation}\label{Eq: Functional valued distribution}
		u:\mathcal{D}(\mathbb{R}^{d+1})\times\mathcal{E}(\mathbb{R}^{d+1})\times\mathcal{E}(\mathbb{R}^{d+1})\to\mathbb{C}\,,\quad (f,\eta,\eta^\prime)\mapsto u(f;\eta,\eta^\prime)\,,
	\end{equation}
which is linear in the first entry and continuous in the locally convex topology of $\mathcal{D}(\mathbb{R}^{d+1})\times\mathcal{E}(\mathbb{R}^{d+1})\times\mathcal{E}(\mathbb{R}^{d+1})$. 
In addition, we indicate the $(k,k^\prime)$-th order functional derivative as the distribution $u^{(k,k^\prime)}\in \mathcal{D}^\prime(\underbrace{\mathbb{R}^{d+1}\times\dots\times\mathbb{R}^{d+1}}_{k+k^\prime+1};\mathrm{Fun}_\mathbb{C})$ such that
\begin{gather}\label{Eq: Functional derivative}
	u^{(k,k^\prime)}(f\otimes\eta_1\otimes\dots\otimes\eta_k\otimes\eta^\prime_1\otimes\dots\otimes\eta^\prime_{k^\prime};\eta,\eta^\prime)\vcentcolon=\notag\\
	\vcentcolon=\left.\frac{\partial^{k+k^\prime}}{\partial s_1\dots\partial s_k\partial s^\prime_1\dots\partial s^\prime_{k^\prime}}u(f;s_1\eta_1+\dots+s_k\eta_k+\eta,s^\prime_1\eta^\prime_1\dots+s^\prime_{k^\prime}\eta_{k^\prime}+\eta^\prime))\right|_{s_1=\dots=s_k=s^\prime_1=\dots=s^\prime_{k^\prime}=0}\,.
\end{gather}
We say that a functional-valued distribution is {\bf polynomial} if $\exists\,(\bar{k},\bar{k}^\prime)\in\mathbb{N}_0\times\mathbb{N}_0$ with $\mathbb{N}_0=\mathbb{N}\cup\{0\}$ such that $u^{(k,k^\prime)}=0$ whenever at least one of these conditions holds true: $k\geq\bar{k}$ or $k^\prime\geq\bar{k}^\prime$. 
The collection of all these functionals is denoted by $\mathcal{D}^\prime(\mathbb{R}^{d+1};\mathsf{Pol}_{\mathbb{C}})$.
\end{definition}

Observe that the subscript $\mathbb{C}$ in $\mathcal{D}^\prime(\mathbb{R}^{d+1};\mathsf{Fun}_{\mathbb{C}})$ and in $\mathcal{D}^\prime(\mathbb{R}^{d+1};\mathsf{Pol}_{\mathbb{C}})$ is here introduced in contrast to the notation of \cite{Dappiaggi:2020gge} to highlight that, since we consider an SPDE with an additive complex white noise, we need to work with a different notion of functional-valued distribution. More precisely Equation \eqref{Eq: Functional valued distribution} codifies that, in addition to the test-function $f$, we need two independent configurations, $\eta,\eta^\prime$ in order to build a functional-valued distribution. The goal is to encode in this framework the information that we want to consider a priori as independent the random distribution $\psi$ and $\bar{\psi}$ in Equation \eqref{Eq: Stochastic NLS}. Their mutual relation via complex conjugation is codified only at the end of our analysis.

\noindent An immediate structural consequence of Definition \ref{Def: Functionals} can be encoded in the following corollary, whose proof is immediate and, therefore we omit it.

\begin{corollary}\label{Cor: Pointwise Algebra Structure}
	 The collection of polynomial functional-valued distributions can be endowed with the structure of a commutative $\mathbb{C}$-algebra such that for all $u,v\in\mathcal{D}^\prime(\mathbb{R}^{d+1};\mathsf{Pol}_{\mathbb{C}})$, and for all $f\in\mathcal{D}(\mathbb{R}^{d+1})$ and $\eta,\eta^\prime\in\mathcal{E}(\mathbb{R}^{d+1})$, 
	 \begin{equation}\label{Eq: Pointwise Product}
	 (uv)(f;\eta,\eta^\prime)=u(f;\eta,\eta^\prime)v(f;\eta,\eta^\prime)\,.
	 \end{equation}
\end{corollary}

A close inspection of Equation \eqref{Eq: Functional derivative} suggests the possibility of introducing a related notion of {\em directional derivative} of a functional $u\in\mathcal{D}^\prime(\mathbb{R}^{d+1};\mathsf{Fun}_{\mathbb{C}})$ by taking an arbitrary but fixed $\zeta\in\mathcal{E}(\mathbb{R}^{d+1})$ and setting for all $(f,\eta)\in\mathcal{D}(\mathbb{R}^{d+1})$,
\begin{equation}\label{Eq: Directional derivative}
\delta_\zeta u(f;\eta,\eta^\prime)\doteq u^{(1,0)}(f\otimes\zeta;\eta,\eta^\prime)\quad\textrm{and}\quad\overline{\delta}_\zeta u(f;\eta,\eta^\prime)\doteq u^{(0,1)}(f\otimes\zeta;\eta,\eta^\prime)\,.
\end{equation}

In order to make Definition \ref{Def: Functionals} more concrete, we list a few basic examples of polynomial functional-valued distributions, which shall play a key r\^{o}le in our investigation, particularly in the construction of the algebraic structures at the heart of our approach.

\begin{example}\label{Ex: Basic Functionals}
For any $f\in\mathcal{D}(\mathbb{R}^{d+1})$ and $\eta,\eta^\prime\in\mathcal{E}(\mathbb{R}^{d+1})$, we call
\begin{align*}
\mathbf{1}(f;\eta,\eta^\prime)\doteq\int_{\mathbb{R}^{d+1}} dx\,f(x),\qquad\Phi(f;\eta,\eta^\prime)=\int_{\mathbb{R}^{d+1}} dx\,f(x)\eta(x),\qquad\overline{\Phi}(f;\eta,\eta^\prime)=\int_{\mathbb{R}^{d+1}} dx\,f(x)\eta^\prime(x)\,,
\end{align*}
where $dx$ is the standard Lesbegue measure on $\mathbb{R}^{d+1}$. Notice that $\Phi$ and $\overline{\Phi}$ are two independently defined functionals and a priori they are not related by complex conjugation as the symbols might suggest. This difference originates from our desire to follow an approach similar to the one often used in quantum field theory, see {\it e.g.} \cite{Araki:1971id}, when analyzing the quantization of spinors. In this case it is convenient to consider a field and its complex conjugate as a priori independent building blocks, since this makes easier the implementation of the canonical anticommutation relations. Here we wish to follow a similar rationale and we decided to keep the symbols $\Phi$ and $\overline{\Phi}$ as a memento that, at the end of the analysis, one has to restore their mutual relation codified by complex conjugation. 
 
In addition, we can construct composite functionals, namely for any $p,q\in\mathbb{N}$, $f\in\mathcal{D}(\mathbb{R}^{d+1})$ and $\eta,\eta^\prime\in\mathcal{E}(\mathbb{R}^{d+1})$, we call 
\begin{align}\label{Eq: generic functional product}
\overline{\Phi}^p\Phi^q(f,\eta,\eta^\prime)=(\Phi^q\overline{\Phi}^p)(f,\eta,\eta^\prime)=\int_{\mathbb{R}^{d+1}} dx\,f(x)\eta^p(x)(\eta^\prime)^q(x)\,.
\end{align}
Observe that, for convenience, we adopt the notation $|\Phi|^{2k}\equiv\Phi^k\overline{\Phi}^k=\overline{\Phi}^k\Phi^k$.

One can readily infer that, for all $p,q\in\mathbb{N}$, $\mathbf{1},\Phi,\overline{\Phi},\overline{\Phi}^p\Phi^q\in\mathcal{D}^\prime(\mathbb{R}^{d+1};\mathsf{Pol}_{\mathbb{C}})$. 
On the one hand all functional derivatives of $\mathbf{1}$ vanish, while, on the other hand 
\begin{align*}
\Phi^{(1,0)}(f\otimes\eta_1;\eta,\eta^\prime)=\overline{\Phi}^{(0,1)}(f\otimes\eta_1;\eta,\eta^\prime)=\int_{\mathbb{R}^{d+1}}dx\,f(x)\eta_1(x)\,,
\end{align*}
whereas $\Phi^{(k,k^\prime)}=0$ for all $k^\prime\neq 0$ or if $k\geq 1$. 
A similar conclusion can be drawn for $\overline{\Phi}$ and for $\overline{\Phi}^p\Phi^q$.
\end{example}

In the spirit of a perturbative analysis of Equation \eqref{Eq: Stochastic NLS} the next step consists of encoding in the polynomial functionals the information that the linear part of the dynamics is ruled by the Schr\"odinger operator $L=i\partial_t+\Delta$. 
To this end, let $u\in\mathcal{D}^\prime(\mathbb{R}^{d+1};\mathsf{Fun}_{\mathbb{C}})$ and let $G$ be the fundamental solution of $L$ whose integral kernel is as per Equation \eqref{Eq: Integral Kernel}. 
Then, for all $\eta,\eta^\prime\in\mathcal{E}(\mathbb{R}^{d+1})$ and for all $f\in\mathcal{D}(\mathbb{R}^{d+1})$, 
\begin{equation}\label{Eq: Action of G}
	(G\circledast u)(f;\eta,\eta^\prime)\doteq u(G\circledast f;\eta,\eta^\prime),
\end{equation}
where, in view of Remark \ref{Rem: la vita e' dura} and of \cite[Th. 8.2.12]{Hormander-I-03} $G\circledast f\in C^\infty(\mathbb{R}^{d+1})$ is such that, for all $h\in\mathcal{D}(\mathbb{R}^{d+1})$, $(G\circledast f)(h)\doteq G(f\otimes h)$. 

We can now collect all the ingredients introduced, building a distinguished commutative algebra. 
We proceed in steps, adapting to the case in hand the procedure outlined in \cite{Dappiaggi:2020gge}. 
As a starting point, we introduce 
\begin{equation}\label{Eq: Pointwise algebra - A_0}
	\mathcal{A}^{\mathbb{C}}_0\doteq\mathcal{E}[\mathbf{1},\Phi,\overline{\Phi}],
\end{equation}
that is the polynomial ring on $\mathcal{E}(\mathbb{R}^{d+1})$ whose generators are the functionals $\mathbf{1},\Phi,\overline{\Phi}$ defined in Example \ref{Ex: Basic Functionals}.
The algebra product is the pointwise one introduced in Corollary \ref{Cor: Pointwise Algebra Structure}. 
Subsequently we encode the action of the fundamental solution of the Schr\"odinger operator as
\begin{subequations}
	\begin{equation}\label{Eq: Pointwise algebra - GA_0}
		G\circledast\mathcal{A}^{\mathbb{C}}_0\doteq\{u\in\mathcal{D}^\prime(\mathbb{R}^{d+1};\mathsf{Fun}_{\mathbb{C}})\;|\;u=G\circledast v,\;\textrm{with}\; v\in\mathcal{A}^{\mathbb{C}}_0\},
	\end{equation}
	\begin{equation}\label{Eq: Pointwise algebra - ccGA_0}
		\overline{G}\circledast\mathcal{A}^{\mathbb{C}}_0\doteq\{u\in\mathcal{D}^\prime(\mathbb{R}^{d+1};\mathsf{Fun}_{\mathbb{C}})\;|\;u=\overline{G}\circledast v,\;\textrm{with}\; v\in\mathcal{A}^{\mathbb{C}}_0\},
	\end{equation}
\end{subequations}
where the action of $G$ is defined in Equation \eqref{Eq: Action of G}. 
Here $\overline{G}$ stands for the complex conjugate of the fundamental solution $G$ as in Equation \eqref{Eq: Integral Kernel}. 
Its action on a function is defined in complete analogy with Equation \eqref{Eq: Action of G}.
In order to account for the possibility of applying to our functionals more than once the fundamental solution $G$, as well as $\overline{G}$, we proceed inductively defining, for every $j\geq 1$
\begin{equation}\label{Eq: Pointwise algebra - A_j}
	\mathcal{A}^{\mathbb{C}}_j=\mathcal{E}[\mathcal{A}^{\mathbb{C}}_{j-1}\cup G\circledast\mathcal{A}^{\mathbb{C}}_{j-1}\cup \overline{G}\circledast\mathcal{A}^{\mathbb{C}}_{j-1}].
\end{equation}  
Since $\mathcal{A}^{\mathbb{C}}_{j-1}\subset\mathcal{A}^{\mathbb{C}}_{j}$ for all $j\geq 1$, the following definition is natural.

\begin{definition}\label{Def: Pointwise Algebra}
	Let $\mathcal{A}^{\mathbb{C}}_j$, $j\geq 0$, be the rings as per Equation \eqref{Eq: Pointwise algebra - A_0} and \eqref{Eq: Pointwise algebra - A_j}. 
	Then we call $\mathcal{A}^{\mathbb{C}}$ the unital, commutative $\mathbb{C}$-algebra obtained as the direct limit 
\begin{align*}	
	\mathcal{A}^{\mathbb{C}}\doteq\varinjlim\mathcal{A}^{\mathbb{C}}_j\,.
\end{align*}
\end{definition}

\begin{remark}\label{moduli2}
	For later convenience, it is important to realize that $\mathcal{A}^{\mathbb{C}}$ ({\em resp.} $\mathcal{A}^{\mathbb{C}}_j$, $j\geq 0$) can be regarded as a graded algebra over $\mathcal{E}(\mathbb{R}^{d+1})$, {\it i.e.}
	\begin{equation}
		\mathcal{A}^{\mathbb{C}}=\bigoplus_{m,l,m',l'\in\mathbb{N}_0}\mathcal{M}_{m,m',l,l'}, \hspace{2cm} \mathcal{A}^{\mathbb{C}}_j=\bigoplus_{m,l,m',l'\in\mathbb{N}_0}\mathcal{M}_{m,m',l,l'}^j\,.
	\end{equation}
	Here $\mathcal{M}_{m,m',l,l'}$ is the $\mathcal{E}(\mathbb{R}^{d+1})$-module generated by those elements in which the fundamental solutions $G$ and $\overline{G}$ act $l$ and $l'$ times, respectively while the overall polynomial degree in $\Phi$ is $m$ and the one in $\overline{\Phi}$ is $m'$. 
	The components of the decomposition of $\mathcal{A}^{\mathbb{C}}_{j}$ are defined as $\mathcal{M}_{m,m',l,l'}^j=\mathcal{M}_{m,m',l,l'}\cap\mathcal{A}^{\mathbb{C}}_j$.
	In the following a relevant r\^{o}le is played by the polynomial degree of an element lying in $\mathcal{A}^{\mathbb{C}}$, ignoring the occurrence of $G$ and $\overline{G}$.
 Therefore we introduce
	\begin{equation}
		\mathcal{M}_{m,m'}\doteq\bigoplus\limits_{\substack{l,l'\in\mathbb{N}_0 \\ p\leq m\\q\leq m'}}\mathcal{M}_{p,q,l,l'},\hspace{2cm} \mathcal{M}^j_{m,m'}\doteq\bigoplus\limits_{\substack{l,l'\in\mathbb{N}_0 \\ p\leq m\\q\leq m'}}\mathcal{M}_{p,q,l,l'}^j\,.
	\end{equation}
	Since it holds that $\mathcal{M}_k\subseteq\mathcal{M}_{k+1}$ for all $k\geq0$, the direct limit is well defined and it holds:
	\begin{equation}
		\mathcal{A}^{\mathbb{C}}=\lim_{m,m'\rightarrow\infty}\mathcal{M}_{m,m'}\,.
	\end{equation}
\end{remark}

\noindent To conclude the section, we establish an estimate on the wavefront set of the derivatives of the functionals lying in $\mathcal{A}^{\mathbb{C}}$, see Definition \ref{Def: Pointwise Algebra}. 
To this end, we fix the necessary preliminary notation, namely, for any $k\in\mathbb{N}$ we set 
\begin{equation}\label{Eq: Delta Diag}
	\delta_{Diag_k}\in\mathcal{D}^\prime(\mathbb{R}^{(d+1)k})\quad|\quad \mathcal{D}(\mathbb{R}^{(d+1)k})\ni f(x_1,\cdots,x_k)\mapsto \delta_{Diag_k}(f)\doteq\int_{\mathbb{R}^{d+1}}dx_1 f(x_1,\cdots,x_1)\,,
\end{equation}
where $x_1,\cdots,x_k\in\mathbb{R}^{d+1}$.

\begin{definition}\label{Def: Cm set}
	For any but fixed $m\in\mathbb{N}$, let us consider any arbitrary partition of the set $\{1,\ldots,m\}$ into the  disjoint union of $p$ non-empty subsets $I_1\uplus\ldots\uplus I_p$, $p$ being arbitrary. 
	Employing the notations  $\widehat{x}_m=(x_1,\ldots,x_m)\in \mathbb{R}^{(d+1)m}$ and $\widehat{t}_{m}=(t_1,\ldots,t_m)\in\mathbb{R}^m$, as well as their counterpart at the level of covectors, we set
	\clacomment{
	\begin{equation}\label{C_m}
		\begin{split}
			C_m:=&\{(\widehat{t}_m, \widehat{x}_m,\widehat{\omega}_m,\widehat{k}_m)\in T^\ast\mathbb{R}^{(d+1)m}\setminus\{0\}\,\vert\,\exists\,l\in\{1,\ldots,m-1\},\\
			&\{1,\ldots,m\}=I_1\uplus\ldots\uplus I_l\text{ such that }\forall i\neq j,\\
			&\forall (a,b)\in I_i\times I_j,\text{ then }t_a\neq t_b\text{ and }\forall j\in\{1,\ldots, l\},\\
			&t_n=t_m\,\forall n,m\in I_j\text{ and }\sum_{m\in I_j}\omega_m=0,\sum_{m\in I_j}k_m=0\}\,,
		\end{split}
	\end{equation}}
where $\vert I_j\vert$ denotes the cardinality of the set $I_j$.
Accordingly we call
	\begin{equation}\label{Eq: Condition_C}
			\mathcal{D}'_C(\mathbb{R}^{(d+1)};\mathsf{Pol}_{\mathbb{C}})\doteq\{u\in\mathcal{D}'(\mathbb{R}^{(d+1)};\mathsf{Pol}_{\mathbb{C}})\;\vert\;\mathrm{WF}(u^{(k,k^\prime)})\subseteq C_{k+k^\prime+1}\,,\;\forall\, k,k^\prime\geq 0\}\,.
	\end{equation}
\end{definition}

\begin{remark}\label{Rem: smooth distributions}
We observe that the elements of the space $\mathcal{D}'_C(\mathbb{R}^{(d+1)};\mathsf{Pol}_{\mathbb{C}})$ are distributions generated by smooth functions, as one can deduce from Equations \eqref{C_m} and \eqref{Eq: Condition_C} setting $k=k^\prime=0$.
\end{remark}

\begin{remark}\label{Rem: stability wrt G}
The space $\mathcal{D}'_C(\mathbb{R}^{(d+1)};\mathsf{Pol}_{\mathbb{C}})$ is stable with respect to the action of the fundamental solution $G$ and of its complex conjugate, more precisely 
\begin{align*}
G\circledast\tau,\,\overline{G}\circledast\tau\in\mathcal{D}'_C(\mathbb{R}^{(d+1)};\mathsf{Pol}_{\mathbb{C}})\,,\qquad\forall\,\tau\in\mathcal{D}'_C(\mathbb{R}^{(d+1)};\mathsf{Pol}_{\mathbb{C}})\,.
\end{align*}
The proof of this property follows the same lines of \cite[Lemma 2.14]{Dappiaggi:2020gge} and thus we omit it. \clacomment{We underline that the only difference concerns the different form of $\mathrm{WF}(G)$, which is nonetheless accounted for by the definition of the sets in Equation \eqref{C_m}.
 
Another important result pertains the estimate of the scaling degree of $G\circledast \tau$. Before dwelling into its calculation, we recall the definition of weighted scaling degree at a point: given $u\in\mathcal{D}'(\mathbb{R}^{d+1})$ and $f\in\mathcal{D}(\mathbb{R}^{d+1})$, consider the scaled function $f_\lambda(t,x)=\lambda^{-(d+2)}f(\lambda^{-2}t,\lambda^{-1}x)$ and $u_\lambda(f)=u(f_\lambda)$.
\begin{align*}
	\mathrm{sd}_{t_0,x_0}(u)=\inf_{\omega^\prime}\{\omega^\prime\in\mathbb{R}\,\vert\,\lim_{\lambda\rightarrow0}\lambda^{\omega^\prime}u_\lambda(t_0,x_0)=0\}.
\end{align*}
In this work we are interested in the scaling degree with respect to the hypersurface $\Lambda_t$, see \cite{Brunetti-Fredenhagen-00}. Under the aforementioned parabolic scaling, the fundamental solution $G$ behaves homogeneously and a direct computation shows that $\mathrm{wsd}_{\Lambda_t}(G)=\mathrm{wsd}_{\Lambda_t}(\overline{G})=d$.}
\clacomment{Having in mind the preceding properties and referring to \cite[App. B]{Dappiaggi:2020gge} and \cite[Lemma 2.14 ]{Dappiaggi:2020gge} for the technical details, it holds
\begin{align*}
\mathrm{wsd}_{\Lambda^{1+k_1+k_2}_t}(G\circledast\tau)^{(k_1,k_2)}\,,\mathrm{wsd}_{\Lambda^{1+k_1+k_2}_t}(\overline{G}\circledast\tau)^{(k_1,k_2)}<\infty\,,
\end{align*}
whenever $\mathrm{wsd}_{\Lambda^{1+k_1+k_2}_t}(\tau)^{(k_1,k_2)}<\infty$. Here $\mathrm{wsd}_{\Lambda^{1+k_1+k_2}_t}$ denotes the scaling degree with respect to the subset
\begin{align*}
	\Lambda^{1+k_1+k_2}_t=\{(t_1,\ldots,t_{1+k_1+k_2};\widehat{x}_1,\ldots,\widehat{x}_{1+k_1+k_2})\,\vert\,t_1=\ldots=t_{1+k_1+k_2}\}\subset\mathbb{R}^{(d+1)(1+k_1+k_2)}.
\end{align*}} 
\end{remark}

\begin{remark}
	\clacomment{We underline that the choice of a parabolic scaling, weighting the time variable twice the spatial ones, is natural since this is the scaling transformation under which the linear part of the Schr\"odinger equation as well as its fundamental solution are scaling invariant.}
\end{remark}

\noindent The following proposition is the main result of this section.

\begin{proposition}\label{Prop: WF of the Algebra}
	Let $\mathcal{A}^{\mathbb{C}}$ be the algebra introduced in Definition \ref{Def: Pointwise Algebra}. 
	In view of Definition \ref{Def: Cm set}, it holds that 
	\begin{align*}
	\mathcal{A}^{\mathbb{C}}\subset\mathcal{D}^\prime_C(\mathbb{R}^{d+1}; \mathsf{Pol}_{\mathbb{C}})\,.
	\end{align*}
\end{proposition}

\begin{proof}
	In view of the characterization of $\mathcal{A}^{\mathbb{C}}$ as an inductive limit, see Definition \ref{Def: Pointwise Algebra}, it suffices to prove the statement for each $\mathcal{A}^{\mathbb{C}}_j$ as in Equation \eqref{Eq: Pointwise algebra - A_0} and \eqref{Eq: Pointwise algebra - A_j}. 
	To this end we proceed inductively on the index $j$.
	
	\vskip .2cm
	
	\noindent{\em Step 1 -- $j=0$:}  In view of Equation \eqref{Eq: Pointwise algebra - A_0}, it suffices to focus the attention on the collection of generators $u_{p,q}\equiv\overline{\Phi}^p\Phi^q$, $p,q\geq 0$, see Example \ref{Ex: Basic Functionals}. 
	If $p=q=0$, there is nothing to prove since we are considering the identity functional $\mathbf{1}$ whose derivatives are all vanishing. 
	Instead, in all other cases, we observe that a direct calculation shows that the derivatives yield either smooth functionals or suitable products between smooth functions and Dirac delta distributions. 
	Per comparison with Definition \ref{Def: Cm set}, we can conclude 
	$$\mathrm{WF}(u_{p,q}^{(k,k^\prime)})\subseteq C_{k+k^\prime+1}\,.$$
	
	\vskip .2cm
	
	\noindent{\em Step 2 -- $j\geq 1$:} As first step, we observe that if $u\in\mathcal{A}^{\mathbb{C}}_j\cap\mathcal{D}^\prime_C(\mathbb{R}^{d+1}; \mathsf{Pol}_{\mathbb{C}})$, then $G\circledast u,\overline{G}\circledast u\in\mathcal{D}^\prime_C(\mathbb{R}^{d+1};\mathsf{Pol}_{\mathbb{C}})$. 
	As a matter of fact, Equation \eqref{Eq: Action of G} entails that, for every $k,k^\prime\geq 0$, $(G\circledast u)^{(k,k^\prime)}=u^{(k,k^\prime)}\cdot (G\otimes 1_{k+k^\prime})$, the dot standing for the pointwise product of distributions, while $1$ stands here for the identity operator. 
	Since $\mathrm{WF}(u^{(k,k^\prime)})\subseteq C_{k+k^\prime+1}$ per hypothesis, it suffices to apply \cite[Lemma 2.14]{Dappiaggi:2020gge}, \clacomment{see Remark \ref{Rem: stability wrt G}}, in combination with Remark \ref{Rem: la vita e' dura} to conclude that $\mathrm{WF}((G\circledast u)^{(k,k^\prime)})\subseteq C_{k+k^\prime+1}$. 
	The same line of reasoning entails that an identical conclusion can be drawn for $\overline{G}$. 
	
	We can now focus on the inductive step. 
	Therefore, let us assume that the statement of the proposition holds true for $\mathcal{A}^{\mathbb{C}}_j$. In view of Equation \eqref{Eq: Pointwise algebra - A_j}, $\mathcal{A}^{\mathbb{C}}_{j+1}$ is generated by $\mathcal{A}^{\mathbb{C}}_j$, $G\circledast\mathcal{A}^{\mathbb{C}}_j$ and $\overline{G}\circledast\mathcal{A}^{\mathbb{C}}_j$. 
	On account of the inductive step, it suffices to focus on the pointwise product of two functionals, say $u$ and $v$ lying in one of the generating algebras. \\
	For definiteness we focus on $u\in\mathcal{A}^{\mathbb{C}}_j$ and $v\in\mathcal{A}^{\mathbb{C}}_j$, the other cases following suit. Using the Leibniz rule, for every $k,k^\prime\geq 0$, it turns out that $(uv)^{(k,k^\prime)}$ is a linear combination with smooth coefficient of products of distributions of the form $u^{(p,p^\prime)}\otimes v^{(q,q^\prime)}$ with $0\leq p\leq k$, $0\leq q\leq k^\prime$ and $p+q=k$ while $p^\prime+q^\prime=k^\prime$. 
	The inductive hypothesis entails
	\begin{align*}
	\mathrm{WF}(u^{(p,p^\prime)})\subseteq C_{p+p^\prime+1},\quad\textrm{and}\quad\mathrm{WF}(v^{(q,q^\prime)})\subseteq C_{q+q^\prime+1}\,.
	\end{align*}
\clacomment{
	A direct inspection of Equation \eqref{C_m} entails that if  $(t,x,\omega_1,\kappa_1,s_{p+p'},y_{p+p^\prime},\omega_{s}, \kappa_y)\in C_{p+p^\prime+1}$ while $(t,x,\omega_2,\kappa_2,r_{q+q'},z_{q+q^\prime},\omega_{r}, \kappa_z)\in C_{q+q^\prime+1}$ then
\begin{align*}	
	(t,x,\omega_1+\omega_2,\kappa_1+\kappa_2,s_{p+p'},y_{p+p^\prime},\omega_s,\kappa_y, r_{q+q'},z_{q+q^\prime},\omega_r,\kappa_z)\in C_{k+k^\prime+1}\,,
	\end{align*}
	from which it descends that $\mathrm{WF}(u^{(p,p^\prime)}\otimes v^{(q,q^\prime)})\subseteq C_{k+k^\prime+1}$.
}
\end{proof}

\section{The Algebra $\mathcal{A}^{\mathbb{C}}_{\cdot_Q}$}\label{Sec: expectations}
The algebra $\mathcal{A}^{\mathbb{C}}$ introduced in Definition \ref{Def: Pointwise Algebra} does not codify the information associated to the stochastic nature of the underlying white noise.

In order to encode such datum in the functional-valued distributions and in the same spirit of \cite{Dappiaggi:2020gge}, we 
introduce a new algebra which is obtained as a \emph{deformation} of the pointwise product of $\mathcal{A}^{\mathbb{C}}$. As it will be manifest in the following discussion, this construction is a priori purely formal unless a suitable renormalization procedure is implemented

\begin{remark}\label{Rem: Q}
As a premise, we introduce the following bi-distribution, constructed out of the fundamental solution $G\in\mathcal{D}^\prime(\mathbb{R}^{d+1}\times\mathbb{R}^{d+1})$ as per Equation \eqref{Eq: Integral Kernel}:
\begin{equation}\label{Eq: definition of Q}
Q\vcentcolon=G\circ G\in\mathcal{D}^\prime(\mathbb{R}^{d+1}\times\mathbb{R}^{d+1})\,,
\end{equation}
where $\circ$ denotes the composition of bi-distributions, namely, for any $f_1,f_2\in\mathcal{D}(\mathbb{R}^{d+1}\times\mathbb{R}^{d+1})$,
\begin{align*}
Q(f_1,f_2)\vcentcolon=(G\circ G)(f_1\otimes f_2)\vcentcolon=[(G\otimes G)\cdot(1_{d+1}\otimes\delta_{\mathrm{Diag}_2}\otimes1_{d+1})](f_1\otimes1_{d+1}\otimes1_{d+1}\otimes f_2)\,.
\end{align*}
Here, with a slight abuse of notation, $1_{d+1}\in\mathcal{D}^\prime(\mathbb{R}^{d+1})$ stands for the distribution generated by the constant smooth function $1$ on $\mathbb{R}^{d+1}$ while $\cdot$ indicates the pointwise product of distributions, see \cite[Thm. 8.2.10]{Hormander-I-03}. Further properties of the composition $\circ$ are discussed in \cite[App. A]{Dappiaggi:2020gge}.

From the perspective of the stochastic process at the heart of Equation \eqref{Eq: Stochastic NLS}, we observe that the bi-distribution $Q$ codifies the covariance of the complex Gaussian random field $\widehat{\varphi}=G\ast\xi$.
More explicitly, it holds 
\begin{equation}\label{Eq: covariance complex Gaussian random field}
\mathbb{E}[\widehat{\varphi}(f_1)\overline{\widehat{\varphi}}(f_2)]=Q(f_1,f_2), \quad\forall f_1,f_2\in\mathcal{D}(\mathbb{R}^{d+1}\times\mathbb{R}^{d+1}).
\end{equation}

Observe, in addition that $\widehat{\varphi}$ can be read also as a solution of the stochastic linear Schr\"odinger equation, namely Equation \eqref{Eq: Stochastic NLS} setting $\lambda=0$.

To conclude our excursus on the bi-distribution $Q$, we highlight two notable properties of its singular structure:
\begin{enumerate}
	\item in view of \cite[Thm. 8.2.14]{Hormander-I-03},
	\begin{align*}
		\mathrm{WF}(Q)\subseteq\mathrm{WF}(G)\,.
	\end{align*}
\item as a consequence of \cite[Lemma B.12]{Dappiaggi:2020gge}, 
$$\mathrm{wsd}_{\clacomment{\Lambda^2_t}}(Q)<\infty,$$
where $\mathrm{wsd}_{\clacomment{\Lambda^2_t}}(Q)$ denotes the scaling degree of the distribution $Q\in\mathcal{D}^\prime(\mathbb{R}^{d+1}\times\mathbb{R}^{d+1})$ with respect to $\Lambda_t^2:=\{(t_1,x_1,t_2,x_2)\in\mathbb{R}^{d+1}\times\mathbb{R}^{d+1}\,\vert\,t_1=t_2\}$. 
\end{enumerate}  
\end{remark}

We are now in a position to introduce the sought deformation which encodes the stochastic properties due to the complex white noise present in Equation \eqref{Eq: Stochastic NLS}. Inspired by \cite{Dappiaggi:2020gge}, as a tentative starting point we set, for any $\tau_1,\tau_2\in\mathcal{A}^\mathbb{C}$, $f\in\mathcal{D}(\mathbb{R}^{d+1})$ and $\eta,\eta^\prime\in\mathcal{E}(\mathbb{R}^{d+1})$,
\begin{equation}\label{Eq: prodotto deformato2}
\begin{split}
    [\tau_1\cdot_Q \tau_2](f;\eta, \eta^\prime):=\sum_{k\geq0}\sum\limits_{\substack{k_1,k_2 \\ k_1+k_2=k}}\frac{1}{k_1!k_2!}\big(\delta_{Diag_2}\otimes Q^{\otimes k_1}\otimes\overline{Q}^{\otimes k_2}\big)
    \cdot\big[\tau_1^{(k_1,k_2)}\Tilde{\otimes}\tau_2^{(k_2,k_
        1)}\big](f\otimes 1_{1+2k};\eta, \eta^\prime)\,.
        \end{split}
\end{equation}
Here with $\tilde{\otimes}$ we denote the tensor product though modified in terms of a permutation of its arguments, which, at the level of integral kernel, reads
\begin{gather}\label{Eq: Tilde otimes}
(\delta_{Diag_2}\otimes Q^{\otimes k_1}\otimes\overline{Q}^{\otimes k_2}\big)
    \cdot\big[\tau_1^{(k_1,k_2)}\Tilde{\otimes}\tau_2^{(k_2,k_
        1)}\big]\notag\\
        =\delta(x_1,x_2)\prod_{j=1}^{k_1}\prod_{\ell=1}^{k_2}Q(z_j,z^\prime_{j})\overline{Q}(y_{\ell},y^\prime_\ell)\tau_1^{(k_1,k_2)}(x_1,z_1,\ldots,z_{k_1},y_{1},\ldots,y_{k_2})\tau_2^{(k_2,k_1)}(x_2,z^\prime_{1},\ldots,z^\prime_{k_1},y^\prime_{1},\ldots,y^\prime_{k_2})\,.
\end{gather}

\begin{remark}\label{Rem: deformed product}
We observe that only a finite number of terms contributes to the sum on the right hand side of Equation \eqref{Eq: prodotto deformato2} on account of the polynomial nature of the functional-valued distributions $\tau_1$ and $\tau_2$.
Nonetheless, at this stage, the right hand side of Equation \eqref{Eq: prodotto deformato2} is only a formal expression, since it can include a priori ill-defined structures such as the coinciding point limit of $Q$ and $\overline{Q}$. In the following theorem we shall bypass this hurdle by means of a renormalization procedure so to give meaning to  Equation \eqref{Eq: prodotto deformato2} for any $\tau_1,\tau_2\in\mathcal{A}^\mathbb{C}$.
\end{remark}

The main motivation for the introduction of a deformation of the algebraic structure is to build an explicit algorithm for computing expectation values and correlation functions of polynomial expressions in the random fields $\widehat{\varphi}=G\ast\xi$ and $\overline{\widehat{\varphi}}=\overline{G}\ast\overline{\xi}$. For this reason. let us illustrate the stochastic interpretation of the deformed product $\cdot_Q$ and its link to the expectation values via the following example.
\begin{example}\label{Ex: Example interpretation dot}
Formally, at the level of integral kernels and referring to the defining properties of the complex white noise in Equation \eqref{Eq: complex white noise}, it holds that
    \begin{equation*}
    \begin{split}
        \mathbb{E}[\widehat{\varphi}^2(f)]&=\mathbb{E}\Big[\int_{(\mathbb{R}^{d+1})^2} dx'dy'\,G(x,x')G(x,y')\xi(x')\xi(y')f(x)\Big]\\
        &=\int_{(\mathbb{R}^{d+1})^2}dx'dy'\,G(x,x')G(x,y')\underbrace{\mathbb{E}\bigl[\xi(x')\xi(y')\bigr]}_{=0}f(x)=0,
    \end{split}
\end{equation*}
for every $f\in\mathcal{D}(\mathbb{R}^{d+1})$. Now let us consider the expectation value of $\vert\widehat{\varphi}\vert^2=\widehat{\varphi}\overline{\widehat{\varphi}}$:
    \begin{equation*}
    \begin{split}
        \mathbb{E}[\widehat{\varphi}\overline{\widehat{\varphi}}(f)]&=\mathbb{E}\Big[\int_{(\mathbb{R}^{d+1})^2} dx'dy'\,G(x,x')\overline{G}(x,y')\xi(x')\overline{\xi}(y')f(x)\Big]\\
        &=\int_{(\mathbb{R}^{d+1})^2}dx'dy'\,G(x,x')\overline{G}(x,y')\underbrace{\mathbb{E}\bigl[\xi(x')\xi(y')\bigr]}_{=\delta(x'-y')}f(x)\\
    &=\int_{\mathbb{R}^{d+1}}dx'\,G(x,x')\overline{G}(x,x')f(x)=(G\circ G^\ast)(f\delta_{\text{Diag}_2})=Q(f\delta_{\text{Diag}_2}),
    \end{split}
\end{equation*}
where we used Equation \eqref{Eq: complex white noise - 2} together with the relation between the fundamental solution of $L$ and its formal adjoint $G^\ast$,  $G^\ast(x,x')=\overline{G(x',x)}$, {\em c.f.} Remark \ref{Rem: la vita e' dura}. On account of Equations \eqref{Eq: prodotto deformato2} and \eqref{Eq: Gammadot on products} we can compute
\begin{equation*}
   (\Phi\cdot_Q\Phi)(f;\eta,\overline{\eta})=\Phi^2(f;\eta,\overline{\eta}),
\end{equation*}
\begin{equation*}
    (\Phi\cdot_Q\overline{\Phi})(f;\eta,\overline{\eta})=\Phi\overline{\Phi}(f;\eta,\overline{\eta})+Q(f\delta_{{Diag}_2}).
\end{equation*}
Observe that, fixing $\eta'=\overline{\eta}$, the two configurations are no longer independent, allowing us to make contact with the stochastic nature of the equation. Evaluating these expressions for $\eta=0$, we obtain
\begin{equation*}
    (\Phi\cdot_Q\Phi)(f;0,0)=0\equiv\mathbb{E}[\widehat{\varphi}^2(f)],\hspace{1cm} (\Phi\cdot_Q\overline{\Phi})(f;0,0)=Q(f\delta_{{Diag}_2})\equiv\mathbb{E}[\widehat{\varphi}\overline{\widehat{\varphi}}(f)].
\end{equation*}
As stated in Remark \ref{Rem: deformed product}, the expressions above are a priori ill-defined, accounting for the singular contribution $Q\delta_{Diag_2}$. This problem is tackled in Theorem \ref{Thm: deformation map cdot}.

This example can be readily extended to arbitrary polynomial expressions of $\varphi$ and $\overline{\varphi}$, highlighting how our guess for the deformed product $\cdot_Q$ codifies properly the expectation values. 
\end{example}

\begin{theorem}\label{Thm: deformation map cdot}
Let $\mathcal{A}^\mathbb{C}$ be the algebra introduced in Definition \ref{Def: Pointwise Algebra} and let $\mathcal{M}_{m,m^\prime}$ be the moduli as per Remark \ref{moduli2}. 
There exists a linear map $\Gamma_{\cdot_Q}^\mathbb{C}:\mathcal{A}^\mathbb{C}\to\mathcal{D}^\prime_C(\mathbb{R}^{d+1};\mathrm{Pol}_\mathbb{C})$ such that:
\begin{enumerate}
\item for any $\tau\in\mathcal{M}_{1,0},\,\mathcal{M}_{0,1}$, it holds
\begin{align}\label{Eq: Gammadot on M_1}
\Gamma_{\cdot_Q}^\mathbb{C}(\tau)=\tau\,;
\end{align}
\item for any $\tau\in\mathcal{A}^\mathbb{C}$, it holds
\begin{align}\label{Eq: Gammadot Gtau}
\nonumber
\Gamma_{\cdot_Q}^\mathbb{C}(G\circledast\tau)=G\circledast\Gamma_{\cdot_Q}^\mathbb{C}(\tau)\,,\\
\Gamma_{\cdot_Q}^\mathbb{C}(\overline{G}\circledast\tau)=\overline{G}\circledast\Gamma_{\cdot_Q}^\mathbb{C}(\tau)\,
\end{align}
\item for any $\tau\in\mathcal{A}^\mathbb{C}$ and $\psi\in\mathcal{E}(\mathbb{R}^{d+1})$, it holds
\begin{align}\label{Eq: Gammadot and derivatives}
\nonumber
\Gamma_{\cdot_Q}^\mathbb{C}\cdot\delta_\psi&=\delta_{\psi}\circ\Gamma_{\cdot_Q}^\mathbb{C}\,,\qquad
&&\Gamma_{\cdot_Q}^\mathbb{C}\cdot\delta_{\overline{\psi}}=\delta_{\overline{\psi}}\circ\Gamma_{\cdot_Q}^\mathbb{C}\,,\\
\Gamma_{\cdot_Q}^\mathbb{C}(\psi\tau)&=\psi\Gamma_{\cdot_Q}^\mathbb{C}(\tau)\,,\qquad
&&\Gamma_{\cdot_Q}^\mathbb{C}(\overline{\psi}\tau)=\overline{\psi}\Gamma_{\cdot_Q}^\mathbb{C}(\tau)\,,
\end{align}
where $\delta_\psi$ denotes the directional functional derivative along $\psi$ as per Equation \eqref{Eq: Directional derivative};
\item denoting by $\sigma_{(p,q)}(\tau)=\mathrm{wsd}_{\mathrm{Diag}_{p+q+1}}(\tau^{(p,q)})$ the weighted scaling degree of $\tau^{(p,q)}$ with respect to the total diagonal of $(\mathbb{R}^{d+1})^{p+q+1}$, see \cite[Rmk. B.9]{Dappiaggi:2020gge}, it holds
\begin{align}\label{Eq: Gammadot and sd}
\sigma_{(p,q)}(\Gamma_{\cdot_Q}^\mathbb{C}(\tau))<\infty\,,\qquad\forall\,\tau\in\mathcal{A}^\mathbb{C}\,.
\end{align}
\end{enumerate}
\end{theorem}

\begin{proof}
\noindent\textbf{Strategy:} The construction of a map $\Gamma_{\cdot_Q}^\mathbb{C}$ satisfying the conditions in the statement of the theorem goes by induction with respect to the indices $m$ and $m^\prime$ in the decomposition 
$\mathcal{A}^\mathbb{C}=\bigoplus_{m,m^\prime\in\mathbb{N}}\mathcal{M}_{m,m^\prime}$, discussed in Remark \ref{moduli2}. Since the proof shares many similarities with the counterpart in \cite{Dappiaggi:2020gge} for the case of a stochastic, scalar, partial differential equation, we shall focus mainly on the different aspects. We start from Equation \eqref{Eq: Gammadot on M_1} and, for $\tau_1,\ldots,\tau_n\in\mathcal{A}^\mathbb{C}$, we set
\begin{align}\label{Eq: Gammadot on products}
\Gamma_{\cdot_Q}^\mathbb{C}(\tau_1\ldots\tau_n)\vcentcolon=\Gamma_{\cdot_Q}^\mathbb{C}(\tau_1)\cdot_Q\ldots\cdot_Q\Gamma_{\cdot_Q}^\mathbb{C}(\tau_n)\,,
\end{align}
where the product $\cdot_Q$ is given by Equation \eqref{Eq: prodotto deformato2}.
As we anticipated above, the right hand side of Equation \eqref{Eq: Gammadot on products} might be only formal since the product $\cdot_Q$ is not \textit{a priori} well defined on the whole $\mathcal{D}^\prime_C(\mathbb{R}^{d+1};\mathrm{Pol}_\mathbb{C})$.\\

\noindent\textbf{The Case $\mathbf{d=1}$:} In this scenario the divergence of $Q\delta_{\text{Diag}_2}$ needs no taming. Working at the level of integral kernels, see Equation \eqref{Eq: Integral Kernel}, the composition $Q=G\circ G$ yields a contribution proportional to $\int_{\mathbb{R}}dt/t\sim \log(t)$, which is locally integrable. 
As a consequence, in such scenario the proof does not need renormalization and it is straightforward starting from the condition in the statement of the theorem as well as from Equation \eqref{Eq: Gammadot on products}.\\

\noindent\textbf{The Case $\mathbf{d\geq2}$: Step 1.} We start by showing that, whenever $\Gamma_{\cdot_Q}^\mathbb{C}$ is well defined for $\tau\in\mathcal{A}^\mathbb{C}$ in such a way in particular that Equations \eqref{Eq: Gammadot and derivatives} and \eqref{Eq: Gammadot and sd} hold true, then the same applies for $G\circledast\tau$ and $\overline{G}\circledast\tau$. We shall only discuss the case of $G\circledast\tau$, the other following suit.

First of all we notice that $\Gamma_{\cdot_Q}^\mathbb{C}(G\circledast\tau)$ is completely defined by Equation \eqref{Eq: Gammadot Gtau}.
In addition, Equation \eqref{Eq: Gammadot and sd} for $\Gamma_{\cdot_Q}^\mathbb{C}(G\circledast\tau)$ is a direct consequence of Remark \ref{Rem: stability wrt G}.
Finally, Equation \eqref{Eq: Gammadot and derivatives} descends from an iteration of the following argument, namely, for any $\eta,\eta^\prime,\zeta\in\mathcal{E}(\mathbb{R}^{d+1})$ and $f\in\mathcal{D}(\mathbb{R}^{d+1})$,
\begin{align*}
\delta_\zeta\circ\Gamma_{\cdot_Q}^\mathbb{C}(G\circledast\tau)(f;\eta,\eta^\prime)&=
\delta_\zeta\circ\Gamma_{\cdot_Q}^\mathbb{C}(\tau)(G\circledast f;\eta,\eta^\prime)=
\Gamma_{\cdot_Q}^\mathbb{C}(\delta_\zeta\tau)(G\circledast f;\eta,\eta^\prime)\\
&=\Gamma_{\cdot_Q}^\mathbb{C}(G\circledast\delta_\zeta\tau)(f;\eta,\eta^\prime)=
\Gamma_{\cdot_Q}^\mathbb{C}(\delta_\zeta G\circledast\tau)(f;\eta,\eta^\prime)\\
&=\Gamma_{\cdot_Q}^\mathbb{C}\circ\delta_\zeta(G\circledast\tau)(f;\eta,\eta^\prime)\,.
\end{align*}
With these data in our hand, we can start the inductive procedure. \\

\noindent\textbf{Step 2: $\mathbf{(m,m^\prime)=(1,1)}$.} As a starting point, we observe that Equation \eqref{Eq: Gammadot on M_1} completely determines the map $\Gamma_{\cdot_Q}^\mathbb{C}$ restricted to the moduli $\mathcal{M}_{(1,0)}$ and $\mathcal{M}_{(0,1)}$. In addition all other required properties hold true automatically. Next, we focus our attention on $\mathcal{M}_{(1,1)}$. In order to consistently construct $\Gamma_{\cdot_Q}^\mathbb{C}$ on $\mathcal{M}_{(1,1)}$ we rely on an inductive argument with respect to the index $j$ subordinated to the decomposition -- see Remark \ref{moduli2}:
\begin{align*}
\mathcal{M}_{(1,1)}=\bigoplus_{j\in\mathbb{N}}\mathcal{M}_{(1,1)}^j\,,\qquad\mathcal{M}_{(1,1)}^j\vcentcolon=\mathcal{M}_{(1,1)}\cap\mathcal{A}_j^\mathbb{C}\,. 
\end{align*}

\noindent Setting $j=0$, we observe that
\begin{align*}
\mathcal{M}_{(1,1)}^0=\mathrm{span}_{\mathcal{E}(\mathbb{R}^{d+1};\mathbb{C})}\{\mathbf{1},\Phi,\overline{\Phi},\Phi\overline{\Phi}\}\,.
\end{align*}
On account of the previous discussion, we are left with the task of defining the action of $\Gamma_{\cdot_Q}^\mathbb{C}$ on $\Phi\overline{\Phi}$. This suffices since $\Gamma_{\cdot_Q}^\mathbb{C}$ can be consequently extended to the whole $\mathcal{M}_{(1,1)}^0$ by linearity.
Recalling that $\Phi\overline{\Phi}$ is defined as in Equation \eqref{Eq: generic functional product}, we set, formally, for any $\eta,\eta^\prime\in\mathcal{E}(\mathbb{R}^{d+1})$ and $f\in\mathcal{D}(\mathbb{R}^{d+1})$,
\begin{align*}
\Gamma_{\cdot_Q}^\mathbb{C}(\Phi\overline{\Phi})(f;\eta,\eta^\prime)\vcentcolon=[\Gamma_{\cdot_Q}^\mathbb{C}(\Phi)\cdot_Q\Gamma_{\cdot_Q}^\mathbb{C}(\overline{\Phi})](f;\eta,\eta^\prime)
=\Phi\overline{\Phi}(f;\eta,\eta^\prime)+(G\cdot\overline{G})(f\otimes1_{d+1})\,,
\end{align*}
where in the last equality we exploited Equation \eqref{Eq: prodotto deformato2} and where $G\cdot\overline{G}$ denotes the pointwise product between $G$ and $\overline{G}$.

We observe that the above formula is purely formal due to the presence of the product $G\cdot\overline{G}$ which is \textit{a priori} ill-defined. As a matter of fact, due to the microlocal behaviour of $G$ codified in Remark \ref{Rem: la vita e' dura} and using H\"ormander criterion for the multiplication of distributions \cite[Th. 8.2.10]{Hormander-I-03}, we can only conclude that
\clacomment{$G\cdot\overline{G}\in\mathcal{D}^\prime(\mathbb{R}^{d+1}\times\mathbb{R}^{d+1}\setminus\Lambda^2_t)$.}
Since $\mathrm{wsd}_{\Lambda^2_t}(G\cdot\overline{G})=2d$, \cite[Thm. B.8]{Dappiaggi:2020gge} and \cite[Rmk. B.9]{Dappiaggi:2020gge} entail the existence of an extension $\widehat{G\cdot\overline{G}}\in\mathcal{D}^\prime(\mathbb{R}^{d+1}\times\mathbb{R}^{d+1})$ of $G\cdot\overline{G}$ which preserves both the scaling degree and the wave-front set. Having chosen once and for all any such extension, we set for all $\eta,\eta^\prime\in\mathcal{E}(\mathbb{R}^{d+1})$ and $f\in\mathcal{D}(\mathbb{R}^{d+1})$,
\begin{align}\label{Eq: renormalized squared modulus}
\Gamma_{\cdot_Q}^\mathbb{C}(\Psi\overline{\Psi})(f;\eta,\eta^\prime)\vcentcolon=\Psi\overline{\Psi}(f;\eta,\eta^\prime)+(\widehat{G\cdot\overline{G}})(f\otimes 1)\,.
\end{align}
To conclude, we observe that Equation \eqref{Eq: renormalized squared modulus} implies Equations \eqref{Eq: Gammadot and derivatives} and \eqref{Eq: Gammadot and sd} by direct inspection, as well as that $\Gamma_{\cdot_Q}^\mathbb{C}(\Psi\overline{\Psi})\in\mathcal{D}^\prime_C(\mathbb{R}^{d+1};\mathrm{Pol}_\mathbb{C})$ on account of Remark \ref{Rem: la vita e' dura}.

Having consistently defined $\Gamma_{\cdot_Q}^\mathbb{C}$ on $\mathcal{M}_{(1,1)}^0$, we assume that the map $\Gamma_{\cdot_Q}^\mathbb{C}$ has been coherently assigned on $\mathcal{M}_{(1,1)}^j$ and we prove the inductive step by constructing it on $\mathcal{M}_{(1,1)}^{j+1}$.
We remark that, given $\tau\in\mathcal{M}_{(1,1)}^{j+1}$, either $\tau=G\circledast\tau^\prime$ with $\tau^\prime\in\mathcal{M}_{(1,1)}^j$ (the case with the complex conjugate $\overline{G}$ is analogous), or $\tau=\tau_1\tau_2$ with $\tau_1\in\mathcal{M}_{(1,0)}^j\cup G\circledast\mathcal{M}_{(1,0)}^j\cup\overline{G}\circledast\mathcal{M}_{(1,0)}^j$ while $\tau_2\in\mathcal{M}_{(0,1)}^j\cup G\circledast\mathcal{M}_{(0,1)}^j\cup\overline{G}\circledast\mathcal{M}_{(0,1)}^j$.

In the first case, $\Gamma_{\cdot_Q}^\mathbb{C}(\tau)$ is defined as per Step \textit{1}. of this proof, using in addition the inductive hypothesis. In the second one, we start from the formal expression which descends from Equation \eqref{Eq: Gammadot on products}. It yields for all $\eta,\eta^\prime\in\mathcal{E}(\mathbb{R}^{d+1})$ and $f\in\mathcal{D}(\mathbb{R}^{d+1})$,
\begin{align*}
\Gamma_{\cdot_Q}^\mathbb{C}(\tau)(f;\eta,\eta^\prime)=
[\Gamma_{\cdot_Q}^\mathbb{C}(\tau_1)\cdot_Q\Gamma_{\cdot_Q}^\mathbb{C}(\tau_2)](f;\eta,\eta^\prime)=
[\Gamma_{\cdot_Q}^\mathbb{C}(\tau_1)\Gamma_{\cdot_Q}^\mathbb{C}(\tau_2)](f;\eta,\eta^\prime)+U(f\otimes1_3)\,,
\end{align*}
where
\begin{align}\label{Eq: ill-defined U}
U(f\otimes1_3)\vcentcolon=[(\delta_{\mathrm{Diag}_2}\otimes Q)\cdot(\Gamma_{\cdot_Q}^\mathbb{C}(\tau_1)^{(1,0)}\Tilde{\otimes}\Gamma_{\cdot_Q}^\mathbb{C}(\tau_2)^{(0,1)})](f\otimes1_3)\,,
\end{align}
and where $\Gamma_{\cdot_Q}^\mathbb{C}(\tau_1)\Gamma_{\cdot_Q}^\mathbb{C}(\tau_2)$ denotes the pointwise product between the functional-valued distributions, which is well-defined on account of Remark \ref{Rem: smooth distributions}.
Yet, Equation \eqref{Eq: ill-defined U} involves the product of singular distributions and it calls for a renormalization procedure.
First of all we observe that, due to the microlocal behaviour of $Q$ -- \textit{cf.} Remark \ref{Rem: Q} --  and due to the inductive hypothesis entailing
\begin{align*}
\mathrm{WF}(\Gamma_{\cdot_Q}^\mathbb{C}(\tau_1)^{(1,0)}),\mathrm{WF}(\Gamma_{\cdot_Q}^\mathbb{C}(\tau_2)^{(0,1)})\subset C_2\,,
\end{align*}
it follows that 
\clacomment{
\begin{align*}
	U\in\mathcal{D}'((\mathbb{R}^{d+1})^4\setminus\Lambda_t^{4,\mathrm{Big}}),
\end{align*}
where $\Lambda_t^{4,\mathrm{Big}}:=\{(\widehat{t},\widehat{x})\in(\mathbb{R}^{d+1})^4\,\vert\,\exists a,b\in\{1,2,3,4\}, a\neq b, t_a=t_b\}$.\\
This estimate can be improved observing that, whenever $(\widehat{t},\widehat{x})\in\Lambda_t^{4,\mathrm{Big}}\setminus\Lambda_t^4$, one of the factors among $(\delta_{\mathrm{Diag}_2}\otimes Q)$, $\Gamma_{\cdot_Q}^\mathbb{C}(\tau_1)^{(1,0)}$ and $\Gamma_{\cdot_Q}^\mathbb{C}(\tau_2)^{(0,1)}$ is smooth and the product of the remaining two is well-defined.
}

\clacomment{As a consequence, $U\in\mathcal{D}^\prime((\mathbb{R}^{d+1})^4\setminus\Lambda_t^4)$. Furthermore \cite[Rmk. B.7]{Dappiaggi:2020gge} here applied to the weighted scaling degree guarantees that
\begin{align*}
\mathrm{wsd}_{\Lambda_t^4}(U)
\leq&\mathrm{wsd}_{\Lambda_t^4}(\delta_{\mathrm{Diag}_2}\otimes Q)+\mathrm{wsd}_{\Lambda_t^4}(\Gamma_{\cdot_Q}^\mathbb{C}(\tau_1)^{(1,0)}\Tilde{\otimes}\Gamma_{\cdot_Q}^\mathbb{C}(\tau_2)^{(0,1)})\\
\leq&\mathrm{wsd}_{\Lambda_t^2}(\delta_{\mathrm{Diag}_2})+\mathrm{wsd}_{\Lambda_t^2}(Q)+\mathrm{wsd}_{\Lambda_t^2}(\Gamma_{\cdot_Q}^\mathbb{C}(\tau_1)^{(1,0)})+\mathrm{wsd}_{\Lambda_t^2}(\Gamma_{\cdot_Q}^\mathbb{C}(\tau_2)^{(0,1)})<\infty\,.
\end{align*}}
Once more, \cite[Thm. B.8]{Dappiaggi:2020gge} and \cite[Rem. B.9]{Dappiaggi:2020gge} entail the existence of an extension $\widehat{U}\in\mathcal{D}^\prime((\mathbb{R}^{d+1})^4)$ of $U\in\mathcal{D}^\prime((\mathbb{R}^{d+1})^4\setminus\Lambda_t^4)$ which preserves both the scaling degree and the wave-front set. Eventually, we set 
\begin{align}\label{Eq: renormalized U}
\Gamma_{\cdot_Q}^\mathbb{C}(\tau)(f;\eta,\eta^\prime)=
[\Gamma_{\cdot_Q}^\mathbb{C}(\tau_1)\Gamma_{\cdot_Q}^\mathbb{C}(\tau_2)](f;\eta,\eta^\prime)+\widehat{U}(f\otimes1_3)\,.
\end{align}
By direct inspection, one can realize that Equation \eqref{Eq: renormalized U} satisfies the conditions codified in Equations \eqref{Eq: Gammadot and derivatives} and \eqref{Eq: Gammadot and sd}. This concludes the proof for the case $(m,m^\prime)=(1,1)$.\\

\noindent\textbf{Step 3: Generic $\mathbf{(m,m^\prime)}$.}
We can now prove the inductive step with respect to the indices $(m,m^\prime)$. In particular, we assume that  $\Gamma_{\cdot_Q}^\mathbb{C}$ has been consistently assigned on $\mathcal{M}_{(m,m^\prime)}$ and we prove that the same holds true for $\mathcal{M}_{(m+1,m^\prime)}$. We stress that one should prove the same statement also for $\mathcal{M}_{(m,m^\prime+1)}$, but since the analysis is mutatis mutandis the same as for $\mathcal{M}_{(m+1,m^\prime)}$, we omit it.

Once more we employ an inductive procedure exploiting that $\mathcal{M}_{(m+1,m^\prime)}=\bigoplus_{j\in\mathbb{N}}\mathcal{M}_{(m+1,m^\prime)}^j$, with $\mathcal{M}_{(m+1,m^\prime)}^j:=\mathcal{M}_{(m+1,m^\prime)}\cap\mathcal{A}_j^\mathbb{C}$. First of all, we focus on the case $j=0$, observing that
\begin{align*}
\mathcal{M}_{(m+1,m^\prime)}^0=\mathrm{span}_{\mathcal{E}(\mathbb{R}^{d+1};\mathbb{C})}\{\mathbf{1},\Phi,\overline{\Phi},\ldots,\Phi^{m+1}\overline{\Phi}^{m^\prime}\}\,.
\end{align*}
On account of the inductive hypothesis, the map $\Gamma_{\cdot_Q}^\mathbb{C}$ has been already assigned on all the generators of $\mathcal{M}_{(m+1,m^\prime)}^0$ but $\Phi^{m+1}\overline{\Phi}^{m^\prime}$.
Hence, to determine $\Gamma_{\cdot_Q}^\mathbb{C}$ on the whole $\mathcal{M}_{(m+1,m^\prime)}^0$, it suffices to establish its action on $\Phi^{m+1}\overline{\Phi}^{m^\prime}$ extending it by linearity.
To this end, we exploit again Equation \eqref{Eq: Gammadot on products} which, on account of the explicit form of Equation \eqref{Eq: prodotto deformato2}, yields
\begin{align*}
[\Gamma_{\cdot_Q}^\mathbb{C}(\Phi^{m+1}\overline{\Phi}^{m^\prime})](f;\eta,\eta^\prime)=\sum_{k=0}^{m+1\wedge m^\prime}k!\binom{m+1}{k}\binom{m^\prime}{k}[Q_{2k}\cdot(\Gamma_{\cdot_Q}^\mathbb{C}(\Phi)^{m+1-k})\Gamma_{\cdot_Q}^\mathbb{C}(\overline{\Phi})^{m^\prime-k}](f;\eta,\eta^\prime)\,,
\end{align*}
where $m+1\wedge m^\prime=\min\{m+1,m^\prime\}$ and where the symbol $Q_{2k}$ is a shortcut notation for
\begin{align*}
Q_{2k}(f)\vcentcolon=(G\overline{G})^{\otimes k}\cdot(\delta_{\mathrm{Diag}_k}\otimes 1_k)(f\otimes 1_{2k-1})\,.\quad\forall f\in\mathcal{D}(\mathbb{R}^{d+1}).
\end{align*}
We observe that, although the distribution $Q_{2k}$ is a priori ill-defined, as discussed in Step 2 of this proof, \clacomment{$G\overline{G}\in\mathcal{D}^\prime((\mathbb{R}^{d+1})^2\setminus\Lambda_t^2)$} admits an extension $\widehat{G\overline{G}}$. In turn this yields a renormalized version of $Q_{2k}$, \textit{i.e.},
\begin{align*}
\widehat{Q}_{2k}(f)\vcentcolon=(\widehat{G\overline{G}})^{\otimes k}\cdot(\delta_{\mathrm{Diag}_k}\otimes 1_k)(f\otimes 1_{2k-1})\,.
\end{align*}
Bearing in mind these premises, we set
\begin{align*}
[\Gamma_{\cdot_Q}^\mathbb{C}(\Phi^{m+1}\overline{\Phi}^{m^\prime})](f;\eta,\eta^\prime)=\sum_{k=0}^{m+1\wedge m^\prime}k!\binom{m+1}{k}\binom{m^\prime}{k}[\widehat{Q}_{2k}\cdot(\Gamma_{\cdot_Q}^\mathbb{C}(\Phi)^{m+1-k})\Gamma_{\cdot_Q}^\mathbb{C}(\overline{\Phi})^{m^\prime-k}](f;\eta,\eta^\prime)\,,
\end{align*}
which is well-defined since it involves only products of distributions generated by smooth functions. All required properties of $\Gamma_{\cdot_Q}^\mathbb{C}$ are satisfied by direct inspection. 

Finally, we discuss the last inductive step, namely we assume that $\Gamma_{\cdot_Q}^\mathbb{C}$ has been assigned on $\mathcal{M}_{(m+1,m^\prime)}^j$ and we construct it consistently on $\mathcal{M}_{(m+1,m^\prime)}^{j+1}$. To this end, let us consider $\tau\in\mathcal{M}_{(m+1,m^\prime)}^{j+1}$. Similarly to the inductive procedure in Step {\em 2.}, on account of the definition of $\mathcal{M}_{(m+1,m^\prime)}^{j+1}$ and of the linearity of $\Gamma_{\cdot_Q}^\mathbb{C}$, it suffices to consider elements $\tau$ either of the form $\tau=G\circledast\tau^\prime$ with $\tau^\prime\in\mathcal{M}_{(m+1,m^\prime)}^{j}$ or of the form
\begin{align*}
\tau=\tau_1\ldots\tau_\ell\,,\qquad\tau_i\in\mathcal{M}_{(m_i,m_i^\prime)}^{j}\cup G\circledast\mathcal{M}_{(m_i,m_i^\prime)}^{j}\cup \overline{G}\circledast\mathcal{M}_{(m_i,m_i^\prime)}^{j}\,,\qquad i\in\{1,\ldots,\ell\}\,,\qquad\ell\in\mathbb{N}\,, 
\end{align*}
where for any $i\in\{1,\ldots,\ell\}$ it holds $m_i,m_i^\prime\in\mathbb{N}\setminus\{0\}$ and $\sum_{i=1}^\ell m_i=m+1$ and $\sum_{i=1}^\ell m^\prime_i=m^\prime$.

We observe that in the first case the sought result descends as a direct consequence of Step {\em 1.} and of the inductive hypothesis. Hence we focus only on the second one.  In view of the inductive hypothesis, $\Gamma_{\cdot_Q}^\mathbb{C}$ has already been assigned on $\tau_i$ for all $i\in\{1,\ldots,\ell\}$.
As before, Equation \eqref{Eq: Gammadot on products} yields a formal expression, namely
\begin{align}
\Gamma_{\cdot_Q}^\mathbb{C}(\tau)(f;\eta,\eta^\prime)=
\sum_{N,M\geq0}\sum_{\substack{N_1+\ldots+N_\ell=N+M\\ M_1+\ldots+M_\ell=N+M}}\mathcal{F}(N_1,\ldots&,N_\ell;M_1,\ldots,M_\ell)\cdot\bigl[(\delta_{Diag_l}\otimes Q^{\otimes N}\otimes\overline{Q}^{\otimes M})\cdot\notag\\
&\Bigl(U_{1}^{(N_1,M_1)}\Tilde{\otimes}\ldots\Tilde{\otimes} U_{\ell}^{(N_\ell,M_\ell)}\Bigr)\bigr](f\otimes\mathbf{1}_{\ell-1+2N+2M};\eta,\eta^\prime)\,,\label{Eq: Gamma long}
\end{align}
where we have introduced the notation $\Gamma_{\cdot_Q}^\mathbb{C}(\tau_i)=U_i$ and where the symbols $\mathcal{F}(N_1,\ldots,N_\ell;M_1,\ldots,M_\ell)$ are $\mathbb{C}$-numbers stemming from the underlying combinatorics. Their explicit form plays no r\^{o}le and hence we omit them. 
First of all, since we consider only polynomial functionals, only a finite number of terms in the above sum is non-vanishing.  Nonetheless, the expression in Equation \eqref{Eq: Gamma long} is still purely formal and renormalization needs to be accounted for. To be more precise, let us introduce
\begin{align*}
    U_{N,M}\vcentcolon=\bigl[(\delta_{Diag_\ell}\otimes Q^{\otimes N}\otimes\overline{Q}^{\otimes M})\cdot\Bigl(U_{1}^{(N_1,M_1)}\Tilde{\otimes}\ldots\Tilde{\otimes} U_{\ell}^{(N_\ell,M_\ell)}\Bigr)\bigr]\,.
\end{align*}
On the one hand, by \cite[Thm. 8.2.9]{Hormander-I-03}, 
\begin{align*}
\mathrm{WF}(\delta_{Diag_\ell}\otimes Q^{\otimes N}\otimes\overline{Q}^{\otimes M})\subseteq\{(\widehat{x}_\ell,\widehat{z}_{2N+2M}&;\widehat{k}_\ell,\widehat{q}_{2N,2M})\in T^\ast (\mathbb{R}^{d+1})^{\ell+2N+2M}\setminus\{0\}\,\vert\, \\
&(\widehat{x}_\ell,\widehat{k}_\ell)\in\mathrm{WF}(\delta_{\text{Diag}_\ell})\,, (\widehat{z}_{2N+2M};\widehat{q}_{2N,2M})\in\mathrm{WF}(Q^{\otimes N}\otimes\overline{Q}^{\otimes M})\}\,,
\end{align*}
while, on the other hand, by the inductive hypothesis, $\mathrm{WF}(U_{i}^{(N_i,M_i)})\subseteq C_{N_i+M_i+1}$ for any $i\in\{1,\ldots,\ell\}$.
These two data together with \cite[Thm. 8.2.10]{Hormander-I-03} imply \clacomment{$U_{N,M}\in\mathcal{D}'(M^{\ell+2N+2M}\setminus\Lambda_t^{\ell+2N+2M,\mathrm{Big}})$}, where
\clacomment{
\begin{align*}
\Lambda_t^{\ell+2N+2M,\mathrm{Big}}\vcentcolon=\{(\widehat{t}_{\ell+2N+2M},\widehat{x}_{\ell+2N+2M})\in\mathbb{R}^{(d+1)(\ell+2N+2M)}\,\vert\,\exists i,j\in\{1,\ldots,\ell+2N+2M\}, t_i=t_j\}\,.
\end{align*}}
In addition, again by \cite[Thm. 8.2.10]{Hormander-I-03}, it holds
 \begin{equation*}
    \begin{split}
        \mathrm{WF}(U_{N,M}&)=\{(\widehat{x}_\ell,\widehat{z}_{2N+2M},\widehat{k}_\ell+\widehat{k'}_{\ell},\widehat{q}_{N_1,M_1}+\widehat{q'}_{N_1,M_1},\ldots,\widehat{q}_{N_\ell,M_\ell}+\widehat{q'}_{N_\ell,M_\ell})\in\\
        &\hspace{1cm}\in T^\ast (\mathbb{R}^{d+1})^{\ell+2N+2M}\setminus\{0\}\,\vert\, (\widehat{x}_\ell,\widehat{k}_\ell)\in\mathrm{WF}(\delta_{\text{Diag}_\ell}),\\
        &\hspace{2cm}\,(\widehat{z}_{2N+2M},\widehat{q}_{2N+2M})\in\mathrm{WF}(Q^{\otimes N}\otimes\overline{Q}^{\otimes M}),\\
        &\hspace{3cm}(x_i,\widehat{z}_{N_i,M_i};k_i,\widehat{q'}_{N_i,M_i})\in C_{1+N_i+M_i} \forall i\in\{1,\ldots\ell\}\}\,.
    \end{split}
\end{equation*}
This estimate can be improved through the following argument, which is a generalization of the one used in Step {\em 2.}
Let $\{A,B\}$ be a partition of the index set $\{1,\ldots,\ell+2N+2M\}$, in two disjoint sets such that if
\clacomment{$\{t_1,x_1,\ldots,t_{\ell+2N+2M},x_{\ell+2N+2M}\}=\{\widehat{t}_A,\widehat{x}_A,\widehat{t}_B,\widehat{x}_B\}$, then $t_a\neq t_b$ for any $t_a\in\widehat{t}_A$ and $t_b\in\widehat{t}_B$.}
In such scenario, the integral kernel of $U_{N,M}$ can be decomposed as\clacomment{
\begin{align*}
U_{N,M}(t_1,x_1,\ldots,t_{\ell+2N+2M},x_{\ell+2N+2M})=K^A_{N,M}(\widehat{t}_A,\widehat{x}_A)S_{N,M}(\widehat{t}_A,\widehat{x}_A,\widehat{t}_B,\widehat{x}_B)K^B_{N,M}(\widehat{t}_B,\widehat{x}_B)\,,
\end{align*}
where the kernel $S_{N,M}(\widehat{t}_A,\widehat{x}_A,\widehat{t}_B,\widehat{x}_B)$ is smooth on such a partition and where $K^A_{N,M}(\widehat{t}_A,\widehat{x}_A)$ and $K^B_{N,M}(\widehat{t}_B,\widehat{x}_B)$ are kernels of distributions appearing in the definition of the map $\Gamma_{\cdot_Q}^\mathbb{C}$ on $\mathcal{M}_{(k,k^\prime)}^p$ with $k<m+1$, $k^\prime<m^\prime$ and $p<j+1$.
This, together with the inductive hypothesis, implies that the product $(K^A_{N,M}\otimes K^B_{N,M})\cdot S_{N,M}$ is well defined.
Since this argument is independent from the partition $\{A,B\}$, it follows that $U_{N,M}\in\mathcal{D}'(M^{\ell+2N+2M}\setminus\Lambda_t^{\ell+2N+2M})$.
To conclude, observing that
\begin{equation*}
    \begin{split}
        \text{wsd}_{\Lambda_t^{\ell+2N+2M}}(U_{N,M})\leq \text{wsd}_{\Lambda_t^{\ell+2N+2M}}(\delta_{Diag_\ell}\otimes Q^{\otimes N}\otimes\overline{Q}^{\otimes M})+\sum_{i=1}^\ell\text{wsd}_{\Lambda_t^{1+N_i+M_i}}(U_{i}^{(N_i,M_i)})<\infty\,,
    \end{split}
\end{equation*}
once more, \cite[Thm. B.8]{Dappiaggi:2020gge} and \cite[Rem. B.9]{Dappiaggi:2020gge} grant the existence of an extension $\widehat{U}_{N,M}\in\mathcal{D}^\prime((\mathbb{R}^{d+1})^{\ell+2N+2M})$ of $U\in\mathcal{D}^\prime((\mathbb{R}^{d+1})^{\ell+2N+2M}\setminus\Lambda_t^{\ell+2N+2M})$ which preserves the scaling degree and the wave-front set.
This allows us to set 
\begin{align}\label{Eq: extension of the general case}
\Gamma_{\cdot_Q}^\mathbb{C}(\tau)(f;\eta,\eta^\prime)=
\sum_{N,M\geq0}\sum_{\substack{N_1+\ldots+N_\ell=N+M\\ M_1+\ldots+M_\ell=N+M}}\mathcal{F}(N_1,\ldots&,N_\ell;M_1,\ldots,M_\ell)\cdot\widehat{U}_{N,M}\,.
\end{align}}
Following the same strategy as in the proof of \cite[Thm. 3.1]{Dappiaggi:2020gge}, \textit{mutatis mutandis}, one can show that Equation \eqref{Eq: extension of the general case} satisfies the conditions required in the statement of the theorem.
\end{proof}

A map $\Gamma_{\cdot_Q}^\mathbb{C}$ satisfying the properties listed in Theorem \ref{Thm: deformation map cdot} is the main ingredient for introducing a deformation of the pointwise product allowing the class of functional-valued distributions to encode the stochastic behaviour induced by the white noise. 
This is codified in the algebra introduced in the following theorem, whose proof we omit since, being completely algebraic, it is analogous to the one of \cite[Cor. 3.3]{Dappiaggi:2020gge}.

\begin{remark}
We observe that the deformation map $\Gamma_{\cdot_Q}^\mathbb{C}$ introduced above is similar to the inverse of the one used to define Wick ordering, see \emph{e.g.}, \cite[Sec.2]{BBS19}.
\end{remark}

\begin{theorem}\label{thm_struttura_deformata}
Let $\Gamma_{\cdot_Q}^\mathbb{C}:\mathcal{A}^\mathbb{C}\rightarrow\mathcal{D}^\prime_C(\mathbb{R}^{d+1};\mathrm{Pol}_\mathbb{C})$ be the linear map introduced via Theorem \ref{Thm: deformation map cdot}.
Moreover, let us set $\mathcal{A}^\mathbb{C}_{\cdot_Q}\vcentcolon=\Gamma_{\cdot_Q}^\mathbb{C}(\mathcal{A})$.
In addition, let
\begin{equation}
    u_1\cdot_{\Gamma_{\cdot_Q}^\mathbb{C}}u_2\vcentcolon=\Gamma_{\cdot_Q}^\mathbb{C}\bigl[{\Gamma_{\cdot_Q}^\mathbb{C}}^{-1}(u_1){\Gamma_{\cdot_Q}^\mathbb{C}}^{-1}(u_2)\bigr],\hspace{1cm}\forall\,u_1,u_2\in\mathcal{A}^\mathbb{C}_{\cdot_Q}\,.
\end{equation}
Then $(\mathcal{A}^\mathbb{C}_{\cdot_Q},\cdot_{\Gamma_{\cdot_Q}^\mathbb{C}})$ is a unital, commutative and associative $\mathbb{C}$-algebra.
\end{theorem}

\paragraph{Uniqueness}
The main result of the previous section consists of the existence of a map $\Gamma_{\cdot_Q}^\mathbb{C}$ satisfying the condition as per Theorem \ref{Thm: deformation map cdot} and which, through a renormalization procedure, codifies at an algebraic level the stochastic properties of $\Psi$ induced by the complex white noise $\xi$.
In the following we shall argue that the map $\Gamma_{\cdot_Q}^\mathbb{C}$ is not unique and, to this end, it is convenient to focus on Equation \eqref{Eq: renormalized squared modulus} and in particular on  $\widehat{G\cdot\overline{G}}\in\mathcal{D}^\prime(\mathbb{R}^{d+1}\times\mathbb{R}^{d+1})$.
\clacomment{As discussed above, this is an extension of  $G\cdot\overline{G}\in\mathcal{D}^\prime(\mathbb{R}^{d+1}\times\mathbb{R}^{d+1}\setminus\Lambda_t^2)$.}
Due to \cite[Thm. B.8]{Dappiaggi:2020gge} and \cite[Rem. B.9]{Dappiaggi:2020gge}, such an extension might not be unique depending on the dimension $d\in\mathbb{N}$ of the underlying space.

In spite of this, we can draw some conclusions on the relation between different prescriptions for the deformation map $\Gamma_{\cdot_Q}^\mathbb{C}$ and in turn for the associated algebras $\mathcal{A}_{\cdot_Q}^\mathbb{C}$.
Results in this direction are analogous to those stated in \cite[Thm. 5.2]{Dappiaggi:2020gge}. Hence mutatis mutandis we adapt them to the complex scenario. Once more we do not give detailed proofs, since these can be readily obtained from the counterparts in \cite{Dappiaggi:2020gge}. 
The following two results characterize completely the arbitrariness in choosing the linear map $\Gamma_{\cdot_Q}^\mathbb{C}$, as well as the ensuing link between the deformed algebras.

\begin{theorem}\label{thm_unicc}
 Let $\Gamma_{\cdot_Q}^\mathbb{C},\,{\Gamma_{\cdot_Q}^\mathbb{C}}':\mathcal{A}^\mathbb{C}\rightarrow\mathcal{D}^\prime_C(\mathbb{R}^{d+1};\mathrm{Pol}_\mathbb{C})$ be two linear maps satisfying the requirements listed in Theorem \ref{Thm: deformation map cdot}. 
 There exists a family $\{C_{\ell,\ell^\prime}\}_{\ell,\ell^\prime\in\mathbb{N}_0}$ of linear maps $C_{\ell,\ell'}:\mathcal{A}^\mathbb{C}\rightarrow\mathcal{M}_{\ell,\ell'}$ satisfying the following properties
 \begin{itemize}
     \item For all $m, m'\in\mathbb{N}$ such that either $m\leq j+1$ or $m'\leq j'+1$, it holds
     \begin{equation*}
         C_{j,j'}[\mathcal{M}_{m,m'}]=0\,,
     \end{equation*}
     \item For all $\ell,\ell'\in\mathbb{N}_0$ and for all $\tau\in\mathcal{A}^\mathbb{C}$, it holds
     \begin{equation*}
         C_{\ell,\ell'}[G\circledast\tau]=G\circledast C_{\ell,\ell'}[\tau],\hspace{1cm}C_{\ell,\ell'}[\overline{G}\circledast\tau]=\overline{G}\circledast C_{\ell,\ell'}[\tau]\,,
     \end{equation*}
     \item For all $\ell,\ell'\in\mathbb{N}$ and for all $\zeta\in\mathcal{E}(\mathbb{R}^{d+1};\mathbb{C})$, it holds
     \begin{equation*}
         \delta_\zeta\circ C_{\ell,\ell'}=C_{\ell-1,\ell'}\circ\delta_\zeta,\hspace{1cm}\delta_{\overline{\zeta}}\circ C_{\ell,\ell'}=C_{\ell,\ell'-1}\circ\delta_{\overline{\zeta}}\,.
     \end{equation*}
     \item For all $\tau\in\mathcal{M}_{m,m'}$, it holds
     \begin{equation*}
         {\Gamma_{\cdot_Q}^\mathbb{C}}'(\tau)=\Gamma_{\cdot_Q}^\mathbb{C}\bigl(\tau+C_{m-1,m'-1}(\tau)\bigr)\,.
     \end{equation*}
 \end{itemize}
\end{theorem}

\begin{corollary}
    Under the hypothesis of Theorem \ref{thm_unicc}, the algebras $\mathcal{A}^\mathbb{C}_{\cdot_Q}=\Gamma_{\cdot_Q}^\mathbb{C}(\mathcal{A}^\mathbb{C})$ and ${\mathcal{A}^\mathbb{C}}'_{\cdot_Q}={\Gamma_{\cdot_Q}^\mathbb{C}}'(\mathcal{A}^\mathbb{C})$, defined as per Theorem \ref{thm_struttura_deformata}, are isomorphic.
\end{corollary}

\section{The Algebra $\mathcal{A}^{\mathbb{C}}_{\bullet_Q}$}\label{Sec: correlations}

The construction of Section \ref{Sec: expectations} allows to compute the expectation values of the stochastic distributions appearing in the \emph{perturbative expansion} of the stochastic nonlinear Schr\"odinger equation, as we shall discuss in Section \ref{Sec: perturbative analysis}.

Alas, this falls short from characterizing, again at a perturbative level, the stochastic behaviour of the perturbative solution on Equation \eqref{Eq: Stochastic NLS}, since it does not allow the computation of multi-local \emph{correlation functions}, which are instead of great relevance for applications. 

Similarly to the case of expectation values, also correlation functions can be obtained through a deformation procedure of a suitable non-local algebra constructed out of $\mathcal{A}^\mathbb{C}_{\cdot_Q}$. Before moving to the main result of this section, we introduce the necessary notation.

As a starting point, we need a multi-local counterpart for the space of polynomial functional-valued distributions, namely
\begin{equation}\label{eq_T'}
      \mathcal{T}_{C}^{'\scriptscriptstyle \mathbb{C}}(\mathbb{R}^{d+1};\text{Pol})\vcentcolon=\mathbb{C}\oplus\bigoplus_{\ell\geq1}\mathcal{D}^\prime_C((\mathbb{R}^{d+1})^\ell;\mathsf{Pol}_{\mathbb{C}})\,.
\end{equation}
We are also interested in 
\begin{equation}\label{eq_modulo2}
    \mathcal{T}^{\scriptscriptstyle \mathbb{C}}(\mathcal{A}^\mathbb{C}_{\cdot_Q})\vcentcolon=\mathcal{E}(\mathbb{R}^{d+1};\mathbb{C})\oplus\bigoplus_{\ell>0}{\mathcal{A}^\mathbb{C}}_{\cdot_Q}^{\,\otimes\ell}\,.
\end{equation}
 
As anticipated, we encode the information on correlation functions through a deformed algebra structure induced over $  \mathcal{T}^{\scriptscriptstyle \mathbb{C}}(\mathcal{A}^\mathbb{C}_{\cdot_Q})$.
The main ingredient is the bi-distribution $Q\in\mathcal{D}^\prime(\mathbb{R}^{d+1}\times\mathbb{R}^{d+1})$ we introduced in Section \ref{Sec: expectations}, which is the $2$-point correlation function of the stochastic convolution $G\circledast\xi$.
Starting from $Q$ we construct the following product $\bullet_Q$: for any $\tau_1\in\mathcal{D}^\prime_C((\mathbb{R}^{d+1})^{\ell_1};\mathsf{Pol}_{\mathbb{C}})$ and $\tau_2\in\mathcal{D}^\prime_C((\mathbb{R}^{d+1})^{\ell_2};\mathsf{Pol}_{\mathbb{C}})$, for any $f_1\in\mathcal{D}((\mathbb{R}^{d+1})^{\ell_1})$, $f_2\in\mathcal{D}((\mathbb{R}^{d+1})^{\ell_2})$ and for all $\eta,\eta^\prime\in\mathcal{E}(\mathbb{R}^{d+1})$ 
\begin{gather}
 (\tau_1\bullet_Q\tau_2)(f_1\otimes f_2;\eta,\eta^\prime)=\notag\\
 \sum\limits_{\substack{k\geq0\\ k_1+k_2=k}}\frac{1}{k_1!k_2!}\big[(\mathbf{1}_{\ell_1+\ell_2}\otimes Q^{\otimes k_1}\otimes\overline{Q}^{\otimes k_2})\cdot(u_1^{(k_1,k_2)}\Tilde{\otimes}u_2^{(k_2,k_1)})\bigr](f_1\otimes f_2\otimes1_{2k};\eta,\eta^\prime)\,.\label{eq_bullet2}
\end{gather}
Here $\Tilde{\otimes}$ codifies a particular permutation of the arguments of the tensor product as in Equation \eqref{Eq: Tilde otimes}.

\begin{theorem}\label{thm_bullet2}
Let $\mathcal{A}^\mathbb{C}_{\cdot_Q}=\Gamma_{\cdot_Q}^\mathbb{C}(\mathcal{A}^\mathbb{C})$ be the deformed algebra defined in Theorem \ref{thm_struttura_deformata} and let us consider the space $ \mathcal{T}_{C}^{'\scriptscriptstyle \mathbb{C}}(\mathbb{R}^{d+1};\mathsf{Pol}_{\mathbb{C}})$ introduced in Equation \eqref{eq_T'}, as well as the universal tensor module $\mathcal{T}^{\scriptscriptstyle \mathbb{C}}(\mathcal{A}^\mathbb{C}_{\cdot_Q})$ defined via Equation \eqref{eq_modulo2}. Then there exists a linear map $\Gamma_{\bullet_Q}^\mathbb{C}:\mathcal{T}^{\scriptscriptstyle\mathbb{C}}(\mathcal{A}^\mathbb{C}_{\cdot_Q})\rightarrow\mathcal{T}'_C(\mathbb{R}^{d+1};\mathsf{Pol}_{\mathbb{C}})$ satisfying the following properties:
\begin{enumerate}
    \item For all $\tau_1\ldots,\tau_\ell\in\mathcal{A}^\mathbb{C}_{\cdot_Q}$ with $\tau_1\in\Gamma_{\cdot_Q}^\mathbb{C}(\mathcal{M}_{1,0})$ or $\tau_1\in\Gamma_{\cdot_Q}^\mathbb{C}(\mathcal{M}_{0,1})$, see Remark \ref{moduli2}, it holds
    \begin{equation}
        \Gamma_{\bullet_Q}^\mathbb{C}(\tau_1\otimes\ldots\otimes\tau_\ell)\vcentcolon=\tau_1\bullet_Q\Gamma_{\bullet_Q}^\mathbb{C}(\tau_2\otimes\ldots\otimes\tau_n)\,,
    \end{equation}
    where $\bullet_Q$ is the product defined via Equation \ref{eq_bullet2}.
    \item For any $\tau_1,\ldots,\tau_\ell\in\mathcal{A}^\mathbb{C}_{\cdot_Q}$, $\eta,\eta^\prime\in\mathcal{E}(M;\mathbb{C})$ and $f_1,\ldots,f_n\in\mathcal{D}(\mathbb{R}^{d+1};\mathbb{C})$ such that there exists $I \subsetneq\{1,\ldots,n\}$ for which
    \begin{equation*}
        \bigcup_{i\in I}\text{supp}(f_i)\cap\bigcup_{j\notin I}\text{supp}(f_j)=\emptyset\,,
    \end{equation*}
    it holds
    \begin{equation*}
        \Gamma_{\bullet_Q}^\mathbb{C}(\tau_1\otimes\ldots\otimes\tau_n)(f_1\otimes\ldots\otimes f_n;\eta,\eta^\prime)
        =\Bigl[\Gamma_{\bullet_Q}^\mathbb{C}\bigl(\bigotimes_{i\in I}\tau_i\bigr)\bullet_Q \Gamma_{\bullet_Q}^\mathbb{C}\bigl(\bigotimes_{j\notin I}\tau_j\bigr)\Bigr](f_1\otimes\ldots\otimes f_n;\eta,\eta^\prime)\,.
    \end{equation*}
    \item Denoting by $\delta_\zeta,\delta_{\overline{\zeta}}$ the functional derivatives in the direction of $\zeta,\overline{\zeta}\in\mathcal{E}(\mathbb{R}^{d+1};\mathbb{C})$, $\Gamma_{\bullet_Q}^\mathbb{C}$ satisfies the following identities:
    \begin{align}
        &\Gamma_{\bullet_Q}^\mathbb{C}(\tau)=\tau\,,\hspace{1cm}\forall\tau\in\mathcal{A}^\mathbb{C}_{\cdot_Q},\label{identita_A}\\
        &\Gamma_{\bullet_Q}^\mathbb{C}\circ\delta_{\zeta}=\delta_\zeta\circ\Gamma_{\bullet_Q}^\mathbb{C}\,,\hspace{0.6cm}\Gamma_{\bullet_Q}^\mathbb{C}\circ\delta_{\overline{\zeta}}=\delta_{\overline{\zeta}}\circ\Gamma_{\bullet_Q}^\mathbb{C}\,,\hspace{0.6cm}\forall\zeta\in\mathcal{E}(\mathbb{R}^{d+1};\mathbb{C})\,,
    \end{align}
    \begin{equation}
    \begin{split}
        \Gamma_{\bullet_Q}^\mathbb{C}(\tau_1&\otimes\ldots\otimes G\circledast\tau_i\otimes\ldots\otimes\tau_n)\\
        &=(\delta_{\text{Diag}_2}^{\otimes i-1}\otimes G\otimes \delta_{\text{Diag}_2}^{\otimes n-i})\circledast\Gamma_{\bullet_Q}^\mathbb{C}(\tau_1\otimes\ldots\otimes\tau_i\otimes\ldots\otimes\tau_n)\,.
        \end{split}
    \end{equation}
     \begin{equation}
    \begin{split}
        \Gamma_{\bullet_Q}^\mathbb{C}(\tau_1&\otimes\ldots\otimes \overline{G}\circledast\tau_i\otimes\ldots\otimes\tau_n)\\
        &=(\delta_{\text{Diag}_2}^{\otimes i-1}\otimes\overline{G}\otimes\delta_{\text{Diag}_2}^{\otimes n-i})\circledast\Gamma_{\bullet_Q}^\mathbb{C}(\tau_1\otimes\ldots\otimes \tau_i\otimes\ldots\otimes\tau_n)\,.
        \end{split}
    \end{equation}
    for all $\tau_1,\ldots,\tau_n\in\mathcal{A}^\mathbb{C}_{\cdot_Q}$, $n\in\mathbb{N}_0$.
\end{enumerate}
\end{theorem}

\noindent Similarly to Section \ref{Sec: expectations}, through a map $\Gamma_{\bullet_Q}^\mathbb{C}$ we can induce an algebra structure.

\begin{theorem}\label{thm_bulletalgebra}
  Let $\Gamma_{\bullet_Q}^\mathbb{C}:\mathcal{T}^{\scriptscriptstyle \mathbb{C}}(\mathcal{A}^\mathbb{C}_{\cdot_Q})\rightarrow\mathcal{T}_C^{'\scriptscriptstyle \mathbb{C}}(\mathbb{R}^{d+1};\mathsf{Pol}_{\mathbb{C}})$ be a linear map satisfying the constraints of Theorem \ref{thm_bullet2} and let us define
    \begin{equation}\label{defoQ}
        \mathcal{A}^\mathbb{C}_{\bullet_Q}\vcentcolon=\Gamma_{\bullet_Q}^\mathbb{C}(\mathcal{A}^\mathbb{C}_{\cdot_Q})\subseteq\mathcal{T}_C^{'\scriptscriptstyle \mathbb{C}}(\mathbb{R}^{d+1};\mathsf{Pol}_{\mathbb{C}})\,.
    \end{equation}
    Furthermore, let us consider the bilinear map $\bullet_{\Gamma_{\bullet_Q}^\mathbb{C}}:\mathcal{A}^\mathbb{C}_{\bullet_Q}\times\mathcal{A}^\mathbb{C}_{\bullet_Q}\rightarrow\mathcal{A}^\mathbb{C}_{\bullet_Q}$ 
    \begin{equation}\label{prod_bullet}
        \tau_1\bullet_{\Gamma_{\bullet_Q}^\mathbb{C}}\tau_2\vcentcolon=\Gamma_{\bullet_Q}^\mathbb{C}\bigl({\Gamma_{\bullet_Q}^\mathbb{C}}^{-1}(u_1)\otimes{\Gamma_{\bullet_Q}^\mathbb{C}}^{-1}(u_1)\bigr)\hspace{1cm}\forall\,\tau_1,\tau_2\in\mathcal{A}^\mathbb{C}_{\bullet_Q}\,.
    \end{equation}
    Then $(\mathcal{A}^\mathbb{C}_{\bullet_Q},\bullet_{\Gamma_{\bullet_Q}^\mathbb{C}})$ identifies a unital, commutative and associative algebra.
\end{theorem}

\begin{remark}
For the sake of brevity, we omit the proofs of Theorems \ref{thm_bullet2} and \ref{thm_bulletalgebra} since, barring minor modifications, they follow the same lines of those outlined in \cite[Sec. 4]{Dappiaggi:2020gge}.
\end{remark}
\begin{example}\label{Ex: Example interpretation bullet}
To better grasp the r\^ole played by this deformation of the algebraic structure, in total analogy with Example \ref{Ex: Example interpretation dot} we compare the two point correlation function of the random field $\varphi$, defined via Equation \eqref{Eq: covariance complex Gaussian random field}, with the functional counterpart
\begin{align*}
       \Gamma_{\bullet_Q}^{\mathbb{C}}(\Phi\overline{\Phi})(f_1\otimes f_2;\eta,\overline{\eta})&=(\Phi\bullet_Q\overline{\Phi})(f_1\otimes f_2;\eta, \overline{\eta})=\Phi\otimes\overline{\Phi}(f_1\otimes f_2;\eta,\overline{\eta})\\
    &+[(\mathbf{1}_2\otimes Q)\cdot(\delta_{Diag_2}\Tilde{\otimes}\delta_{Diag_2})](f_1\otimes f_2\otimes 1_2;\eta, \overline{\eta})\\
    &=\Phi\otimes\overline{\Phi}(f_1\otimes f_2;\eta,\overline{\eta})+Q(f_1\otimes f_2).
\end{align*}
which, evaluated at the configurations $\eta=0$, yields $(\Phi\bullet_Q\overline{\Phi})(f_1\otimes f_2;0,0)=Q(f_1\otimes f_2)\equiv\omega_2(f_1\otimes f_2)$.

Note that we could have also computed the correlation function of the random field $\varphi$ with itself, obtaining zero. More generally, the two-point correlation function of a polynomial expression of $\varphi$ and $\overline{\varphi}$ with itself turns out to be trivial, while the one of the expression and its complex conjugate contains relevant stochastic information, as shown in Section \ref{Sec: perturbative analysis}.
\end{example}

Similarly to the case of the deformation map $\Gamma_{\cdot_Q}^\mathbb{C}$ and the algebra $\mathcal{A}^\mathbb{C}_{\cdot_Q}$, also for $\Gamma_{\bullet_Q}^\mathbb{C}$ and the algebra $\mathcal{A}^\mathbb{C}_{\bullet_Q}$ one can discuss the issue of uniqueness due to the renormalization procedure exploited in the construction of $\Gamma_{\bullet_Q}^\mathbb{C}$. The following theorem deals with this hurdle. Once more we do not give a detailed proof, since it can be readily obtained from the counterpart in \cite{Dappiaggi:2020gge}. 

\begin{theorem}\label{Thm: Uniqueness bullet algebra}
Let $\Gamma_{\bullet_Q}^{\mathbb{C}}\,{\Gamma_{\bullet_Q}^{\mathbb{C}}}':\mathcal{A}^\mathbb{C}_{\cdot_Q}\rightarrow\mathcal{T}^{'\scriptscriptstyle \mathbb{C}}_C(\mathbb{R}^{d+1};\mathsf{Pol}_{\mathbb{C}})$ be two linear maps satisfying the requirements listed in Theorem \ref{thm_bullet2}. Then there exists a family $\{C_{\underline{m},\underline{m}'}\}_{\underline{m},\underline{m}'\in\mathbb{N}_0^{\mathbb{N}_0}}$ of linear maps $C_{\underline{m},\underline{m}'}:\mathcal{T}^{\scriptscriptstyle \mathbb{C}}(\mathcal{A}^\mathbb{C}_{\cdot_Q})\rightarrow\mathcal{T}^{\scriptscriptstyle \mathbb{C}}(\mathcal{A}^\mathbb{C}_{\cdot_Q})$, the space $\mathcal{T}^{\scriptscriptstyle \mathbb{C}}(\mathcal{A}^\mathbb{C}_{\cdot_Q})$ being defined via Equation \eqref{eq_modulo2}, satisfying the following properties:
\begin{enumerate}
    \item For all $j\in\mathbb{N}_0$, it holds
    \begin{equation*}
        C_{\underline{m},\,\underline{m}'}[(\mathcal{\mathcal{A}^\mathbb{C}_{\cdot_Q}})^{\otimes j}]\subseteq\mathcal{M}_{m_1,m_1'}\otimes\ldots\otimes\mathcal{M}_{m_j,m_j'},
    \end{equation*}
    while, if either $m_i\leq l_i-1$ or $m'_i\leq l'_i-1$ for some $i\in\{1,\ldots, j\}$, then
    \begin{equation*}
        C_{\underline{l},\,\underline{l}'}[\mathcal{M}_{m_1,m_1'}\otimes\ldots\otimes\mathcal{M}_{m_j,m'_j}]=0.
    \end{equation*}
    \item For all $j\in\mathbb{N}\cup\{0\}$ and $u_1\ldots, u_j\in\mathcal{A}^\mathbb{C}_{\cdot_Q}$, the following identities hold true:
    \begin{equation*}
    \begin{split}
        C_{\underline{m},\,\underline{m}'}[u_1\otimes\ldots&\otimes G\circledast u_k\otimes\ldots\otimes u_j]\\
        &=(\delta_{\text{Diag}_2}^{\otimes(k-1)}\otimes G\otimes\delta_{\text{Diag}_2}^{\otimes(j-k)})\circledast C_{\underline{m},\,\underline{m}'}[u_1\otimes\ldots\otimes u_j],
    \end{split}
    \end{equation*}
      \begin{equation*}
    \begin{split}
        C_{\underline{m},\,\underline{m}'}[u_1\otimes\ldots&\otimes \overline{G}\circledast u_k\otimes\ldots\otimes u_j]\\
        &=(\delta_{\text{Diag}_2}^{\otimes(k-1)}\otimes \overline{G}\otimes\delta_{\text{Diag}_2}^{\otimes(j-k)})\circledast C_{\underline{m},\,\underline{m}'}[u_1\otimes\ldots\otimes u_j],
    \end{split}
    \end{equation*}
    \begin{equation*}
        \delta_\psi C_{\underline{m},\,\underline{m}'}[u_1\otimes\ldots\otimes u_j]=\sum_{a=1}^jC_{\underline{m}(a),\,\underline{m}'}[u_1\otimes\ldots\otimes\delta_\psi u_a\otimes\ldots\otimes u_j],
    \end{equation*}
       \begin{equation*}
        \delta_{\overline{\psi}} C_{\underline{m},\,\underline{m}'}[u_1\otimes\ldots\otimes u_j]=\sum_{a=1}^jC_{\underline{m},\,\underline{m}'(a)}[u_1\otimes\ldots\otimes\delta_{\overline{\psi}} u_a\otimes\ldots\otimes u_j],
    \end{equation*}
    where $\underline{m}(a)_i=m_i$ if $i\neq a$ and $\underline{m}(a)_a=m_a-1$, while $\delta_\psi$, $\delta_{\overline{\psi}}$ are the directional derivatives with respect to $\psi,\overline{\psi}\in\mathcal{E}(\mathbb{R}^{d+1};\mathbb{C})$, see Definition \ref{Def: Functionals}. Here $G\in\mathcal{D}^\prime(\mathbb{R}^{d+1}\times \mathbb{R}^{d+1})$ is a fundamental solution of the parabolic operator $L$.
    \item Let us consider two linear maps $\Gamma_{\cdot_Q}^\mathbb{C},\,{\Gamma_{\cdot_Q}^\mathbb{C}}':\mathcal{A}^\mathbb{C}\rightarrow\mathcal{D}'_\mathbb{C}(\mathbb{R}^{d+1};\mathsf{Pol}_{\mathbb{C}})$ compatible with the constraints of Theorem \ref{Thm: deformation map cdot}. For all $u_{m_1,m'_1},\ldots, u_{m_j,m'_j}\in\mathcal{A}^\mathbb{C}$ with $u_{m_i,m'_i}\in\mathcal{M}_{m_i,m'_i}$ for all $i\in\{1,\ldots,j\}$ and for $f_1,\ldots,f_j\in\mathcal{D}(\mathbb{R}^{d+1})$ it holds
    \begin{equation*}
        \begin{split}
            {\Gamma_{\bullet_Q}^\mathbb{C}}'(&{\Gamma_{\cdot_Q}^\mathbb{C}}^{'\,\otimes j}(u_{m_1,m'_1}\otimes\ldots\otimes u_{m_j,m'_j}))(f_1\otimes\ldots\otimes f_j)\\
            &=\Gamma_{\bullet_Q}^\mathbb{C}({\Gamma_{\cdot_Q}^\mathbb{C}}^{\otimes j}(u_{m_1,m_1'}\otimes\ldots\otimes u_{m_j,m'_j}))(f_1\otimes\ldots\otimes f_j)\\
            &+\sum_{\mathcal{G}\in \mathcal{P}(1,\ldots,j)}\Gamma_{\bullet_Q}^\mathbb{C}\Bigl[{\Gamma_{\cdot_Q}^\mathbb{C}}^{\otimes\vert\mathcal{G}\vert}C_{\underline{m}_{\mathcal{G}},\,\underline{m}'_{\mathcal{G}}}\Bigl(\bigotimes_{I\in\mathcal{G}}\prod_{i\in I}u_{m_i,\,m'_i}\Bigr)\Bigr]\Bigl(\bigotimes_{I\in\mathcal{G}}\prod_{i\in I}f_i\Bigr).
        \end{split}
    \end{equation*}
    Here $\mathcal{P}(1,\ldots,j)$ denotes the set of all possible partitions of $\{1,\ldots,j\}$ into non-empty disjoint subsets while $\underline{m}_{\mathcal{G}}=(m_I)_{I\in\mathcal{G}}$, where $m_I:=\sum_{i\in I}m_i$.
\end{enumerate}
\end{theorem}

\section{Perturbative Analysis of the Nonlinear Schr\"odinger Dynamics}\label{Sec: perturbative analysis}

In the following we apply the framework devised in Sections \ref{Sec: expectations} and \ref{Sec: correlations} to the study of the stochastic nonlinear Schr\"odinger equation, that is Equation \eqref{Eq: Stochastic NLS}, where we set for simplicity $\kappa=1$. 

The first step consists of translating the equation of interest to a functional formalism, replacing the unknown random distribution $\psi$ with a polynomial functional-valued distribution $\Psi\in\mathcal{D}^\prime(\mathbb{R}^{d+1};\mathsf{Pol}_{\mathbb{C}})$ as per Definition \ref{Def: Functionals}. Recalling that we denote by $G\in\mathcal{D}'(\mathbb{R}^{d+1}\times \mathbb{R}^{d+1})$ the fundamental solution of the Schr\"odinger operator, see Equation \eqref{Eq: Integral Kernel}, the integral form of Equation \eqref{Eq: Stochastic NLS} reads
\begin{equation}\label{Eq: Functional equation}
	\Psi= \Phi+\lambda G\circledast\overline{\Psi}\Psi^2,\hspace{1cm}\lambda\in\mathbb{R}_+,
\end{equation}
where $\Phi\in\mathcal{D}^\prime(\mathbb{R}^{d+1};\mathsf{Pol}_{\mathbb{C}})$ is the functional defined in Example \ref{Ex: Basic Functionals}.

Equation \eqref{Eq: Functional equation} cannot be solved exactly and therefore we rely on a perturbative analysis. Hence we expand the solution as a formal power series in the coupling constant $\lambda\in\mathbb{R}_{+}$ with coefficients lying in $\mathcal{A}^{\mathbb{C}}\subset\mathcal{D}^\prime_C(\mathbb{R}^{d+1};\mathsf{Pol}_{\mathbb{C}})$, see Definition \ref{Def: Pointwise Algebra}:
\begin{equation}\label{Eq: Perturbative expansion}
	\Psi\llbracket\lambda\rrbracket=\sum_{k\geq0}\lambda^{k}F_{k},\hspace{1cm}F_k\in\mathcal{A}^{\mathbb{C}}.
\end{equation}
The coefficients of Equation \eqref{Eq: Perturbative expansion} are determined through an iteration procedure yielding, at first orders
\begin{equation}\label{Eq: Espressione coefficienti}
	\begin{split}
		&F_0=\Phi\,,\\
		&F_1=G\circledast\overline{\Phi}\Phi^2\,,\\
		&F_2=G\circledast(\overline{F_1}F_0^2+2\overline{F_0}F_0F_1)\,,
	\end{split}
\end{equation}
while the $k$-th order contribution reads
\begin{equation}\label{Eq: Generic coefficient}
	F_k=\sum\limits_{\substack{k_1,k_2,k_3\in\mathbb{N}_0 \\ k_1+k_2+k_3=k-1}} G\circledast(\overline{F}_{k_1}F_{k_2}F_{k_3}),\hspace{1cm}k\geq3\,.
\end{equation}
We also underline that the $k$-th coefficient is constructed out of lower order coefficients.
\begin{remark}
	Working on $\mathbb{R}^{d+1}$ we are forced to introduce a cut-off function $\chi\in\mathcal{D}(\mathbb{R}^{d+1})$ in order to implement the necessary support conditions yielding a well-defined convolution $G\circledast u$, see \cite[Chap. 4]{Hormander-I-03}. As per Remark \ref{Rem: Cut-Off on G}, henceforth the cut-off is left implied.
\end{remark}
The perturbative expansion of Equation \eqref{Eq: Perturbative expansion} is particularly useful for computing statistical quantities at every order in perturbation theory by means of the deformation maps, as outlined in the preceding sections. In what follows we explicitly calculate the expectation value of the solution at first order in $\lambda$, as well as its two point correlation function.
\vskip .3cm
\noindent\textbf{\normalsize Expectation value.}

The strategy is thus the following: we first construct the formal perturbative solution $\Psi\llbracket\lambda\rrbracket$ in the algebra $\mathcal{A}\llbracket\lambda\rrbracket$ and then, on account of Theorem \ref{thm_struttura_deformata} and of the stochastic interpretation of the map $\Gamma_{\cdot_Q}^{\mathbb{C}}$ introduced in Theorem \ref{Thm: deformation map cdot}, we consider $\Gamma_{\cdot_Q}^{\mathbb{C}}(\Psi\llbracket\lambda\rrbracket)$.
Bearing in mind this comment we consider
\begin{equation}\label{Eq: perturbative deformed solution}
	\Psi_{\cdot_Q}\llbracket\lambda\rrbracket\vcentcolon=\Gamma_{\cdot_Q}^{\mathbb{C}}(\Psi\llbracket\lambda\rrbracket)=\sum_{k}\lambda^k\Gamma_{\cdot_Q}^{\mathbb{C}}(F_k)\in\mathcal{A}_{\cdot_Q}^{\mathbb{C}}\llbracket\lambda\rrbracket\,,
\end{equation}
where we exploited Equation \eqref{Eq: Perturbative expansion} as well as the linearity of $\Gamma_{\cdot_Q}^{\mathbb{C}}$. 

As outlined in Example \ref{Ex: Example interpretation dot}, the action of $\Gamma_{\cdot_Q}^{\mathbb{C}}$ on elements lying in $\mathcal{A}^{\mathbb{C}}$  coincides with the expectation value of the equivalent expression in terms of the random field $\varphi$. In view of Equation \eqref{Eq: perturbative deformed solution} up to the first order in perturbation theory, it holds
\begin{equation*}
	\Psi\llbracket\lambda\rrbracket=\Phi+\lambda G\circledast\overline{\Phi}\Phi^2+ \mathcal{O}(\lambda^2)\,.
\end{equation*}
Hence, taking into account the defining properties of $\Gamma_{\cdot_Q}^{\mathbb{C}}$ encoded in Equations \eqref{Eq: Gammadot on M_1} and \eqref{Eq: Gammadot Gtau}, the $\eta$-shifted expectation value of the solution reads
\begin{equation*}
	\begin{split}
		\mathbb{E}[\psi_\eta\llbracket\lambda\rrbracket(f)]&=\Gamma_{\cdot_Q}^{\mathbb{C}}(\Psi\llbracket\lambda\rrbracket)(f;\eta,\overline{\eta})=\Gamma_{\cdot_Q}^{\mathbb{C}}(\Phi)(f;\eta,\overline{\eta})+\lambda\Gamma_{\cdot_Q}^{\mathbb{C}}(G\circledast\overline{\Phi}\Phi^2)(f;\eta,\overline{\eta})+\mathcal{O}(\lambda^2)\\
		&=\Phi(f;\eta,\overline{\eta})+\lambda G\circledast\Gamma_{\cdot_Q}^{\mathbb{C}}(\overline{\Phi}\Phi^2)(f;\eta,\overline{\eta})+\mathcal{O}(\lambda^2)\,,
	\end{split}
\end{equation*}
for all $f\in\mathcal{D}(\mathbb{R}^{d+1}), \eta\in\mathcal{E}(\mathbb{R}^{d+1})$. Notice that, at this stage, we have restored the information that $\Phi$ and $\overline{\Phi}$ are related by complex conjugation by choosing $\eta$ and $\bar{\eta}$ as the underlining configurations. The explicit expression of $\cdot_Q$ in Equation \eqref{Eq: prodotto deformato2} entails
\begin{equation*}
	\begin{split}
		\Gamma_{\cdot_Q}^{\mathbb{C}}(\overline{\Phi}\Phi^2)(f;\eta,\overline{\eta})&=(\overline{\Phi}\cdot_Q\Gamma_{\cdot_Q}^{\mathbb{C}}(\Phi^2))(f;\eta,\overline{\eta})=\overline{\Phi}\Phi^2(f;\eta,\overline{\eta})\\
		&+2(\delta_{Diag_2}\otimes\overline{Q})(\delta_{Diag_2}\otimes\Phi\delta_{Diag_2})(f\otimes\delta_{Diag_3};\eta,\overline{\eta})\\
		&=\overline{\Phi}\Phi^2(f;\eta,\overline{\eta})+2\overline{C}\Phi(f;\eta,\overline{\eta})\,,
	\end{split}
\end{equation*}
where $\overline{C}\in\mathcal{E}(\mathbb{R}^{d+1})$ is a smooth function whose integral kernel reads $\overline{C}(x)=\chi(x)(\widehat{G\cdot\overline{G}})(\delta_x\otimes\chi)$. Here $\widehat{G\cdot\overline{G}}$ is any but fixed extension of \clacomment{$G\cdot\overline{G}\in\mathcal{D}'(\mathbb{R}^{d+1}\times\mathbb{R}^{d+1}\setminus\Lambda_t^2)$} to the whole $\mathbb{R}^{(d+1)}\times\mathbb{R}^{(d+1)}$, see the proof of Theorem \ref{Thm: deformation map cdot}. We adopted the notation $\psi_\eta$ to underline that its expectation value is in terms of the shifted random distribution $\varphi_\eta=G\circledast\xi+\eta$, which is still Gaussian, but with expectation value $\mathbb{E}[\varphi_\eta]=\eta$. To wit, the centered Gaussian noise properties can be recollected setting $\eta=0$. Thus the expectation value of the solution up to order $\mathcal{O}(\lambda^2)$ reads
\begin{equation}
	\mathbb{E}[\psi_0\llbracket\lambda\rrbracket]=\mathcal{O}(\lambda^2)\,.
\end{equation}
This result is a direct consequence of the cubic nonlinear term and it can be extended to any order in perturbation theory.
\begin{theorem}\label{Thm: valor medio nullo}
	Let $\Psi\llbracket\lambda\rrbracket\in\mathcal{A}^\mathbb{C}$ be the perturbative solution of Equation \eqref{Eq: Functional equation}. Then
	\begin{equation*}
		\mathbb{E}[\psi_0\llbracket\lambda\rrbracket(f)]=\Gamma_{\cdot_Q}^\mathbb{C}(\Psi\llbracket\lambda\rrbracket)(f;0,0)=0,\hspace{1cm}\forall f\in\mathcal{D}(\mathbb{R}^{d+1})\,.
	\end{equation*}
\end{theorem}
\begin{proof}
	By direct inspection of Equations \eqref{Eq: prodotto deformato2} and \eqref{Eq: Gammadot on products} one can infer that the action of $\Gamma_{\cdot_Q}^\mathbb{C}$ on $u\in\mathcal{A}^{\mathbb{C}}$ yields a sum of terms whose polynomial degree in $\Phi$ and $\overline{\Phi}$ is decreased by an even number due to contractions between pairs of fields $\Phi$ and $\overline{\Phi}$. Hence it is sufficient to show that all perturbative coefficients $F_k$ of Equation \eqref{Eq: Perturbative expansion} lie in $\mathcal{M}_{m,m'}$ with $m+m'=2l+1$ for $l\in\mathbb{N}\setminus\{0\}$. If this holds true, all terms resulting from the action of $\Gamma_{\cdot_Q}^\mathbb{C}(u)$ would contain at least one field $\Phi$ or $\overline{\Phi}$ and their evaluation at $\eta=0$ would vanish.
	
	Let us define the vector space $\mathcal{O}\subset\mathcal{A}^\mathbb{C}$ of polynomial functional-valued distributions of odd polynomial degree, namely $u^{(k_1,k_2)}(\cdot,0)=0$ for all $u\in\mathcal{O}$ with
	$k_1,k_2\in\mathbb{N}_0$ such that $k_1+k_2=2l$, $l\in\mathbb{N}$, see Definition \ref{Def: Functionals}. Observe that $\mathcal{O}$ is closed under the action of $\Gamma_{\cdot_Q}^\mathbb{C}$. In addition it holds that $u_1u_2u_3\in\mathcal{O}$ for any $u_1,u_2,u_3\in\mathcal{O}$ and if $u\in\mathcal{O}$, then $\overline{u}\in\mathcal{O}$.
	
	Focusing on the problem under scrutiny, the thesis translates into requiring that $\Psi\llbracket\lambda\rrbracket\in\mathcal{O}$, which via Equation \eqref{Eq: Perturbative expansion} is equivalent to $F_k\in\mathcal{O}$ for all $k\in\mathbb{N}$. 
	Proceeding inductively, it holds true that, for $k=0$, $F_{0}=\Phi\in\mathcal{O}$. 
	Assuming that $F_j\in\mathcal{O}$ for any $j\leq k\in\mathbb{N}$, our goal is to prove that the same holds true for $F_{k+1}$. 
	Invoking Equation \eqref{Eq: Generic coefficient}, it holds
	\begin{equation*}
		F_{k+1}=\sum\limits_{\substack{k_1,k_2k_3\in\mathbb{N} \\ k_1+k_2+k_3=k}}G\circledast(\overline{F}_{k_1}F_{k_2}F_{k_3})\,.
	\end{equation*}
	Since $k_i\leq k$ for all $i\in\{1,2,3\}$, the inductive hypothesis entails that $F_{k_i}\in\mathcal{O}$. 
	From the aforementioned properties of $\mathcal{O}$ we conclude that $F_{k+1}\in\mathcal{O}$.
\end{proof}
\begin{remark}
	Theorem \ref{Thm: valor medio nullo} holds true for any nonlinear term of the form $\vert\psi\vert^{2\kappa}\psi=\overline{\psi}^\kappa\psi^{\kappa+1}$, $\kappa\in\mathbb{N}$ in Equation \eqref{Eq: Stochastic NLS} since the proof relies on the property that the perturbative coefficients contain an odd number of fields. 
\end{remark}
\vskip .3cm
\noindent\textbf{\normalsize Two-point correlation function.} Another crucial information encoded in the solution of an SPDE resides in the correlation functions which describe the non-local behaviour of the system. Being interested in the calculation of the two-point correlation function, we focus on the deformed algebra $\mathcal{A}^\mathbb{C}_{\bullet_Q}=\Gamma^\mathbb{C}_{\bullet_Q}(\mathcal{A}_{\cdot_Q}^\mathbb{C})$, see Theorems \ref{thm_bullet2} and \ref{thm_bulletalgebra}. In view of Example \ref{Ex: Example interpretation bullet}, the two-point correlation function of the perturbative solution $\psi_\eta\llbracket\lambda\rrbracket$, $\eta\in\mathcal{E}(\mathbb{R}^{d+1})$, localized via $f_1,f_2\in\mathcal{D}(\mathbb{R}^{d+1})$, reads
\begin{equation}\label{Eq: due punti}
	\begin{split}
		\omega_2(\psi_\eta\llbracket\lambda\rrbracket;f_1&\otimes f_2)=\bigl(\Psi_{\cdot_Q}\llbracket\lambda\rrbracket\bullet_{\Gamma_{\bullet_Q}^\mathbb{C}}\overline{\Psi}_{\cdot_Q}\llbracket\lambda\rrbracket\bigr)(f_1\otimes f_2;\eta,\overline{\eta})\\
		&=\Gamma_{\bullet_Q}^\mathbb{C}\bigl[\Gamma_{\cdot_Q}^\mathbb{C}(\Psi\llbracket\lambda\rrbracket)\otimes\Gamma_{\cdot_Q}^{\mathbb{C}}(\overline{\Psi}\llbracket\lambda\rrbracket)\bigr](f_1\otimes f_2;\eta;\overline{\eta})\\
		&=\sum_{k\geq0}\lambda^k\sum\limits_{\substack{k_1,k_2 \\ k_1+k_2=k}}\Gamma_{\bullet_Q}\big(\Gamma_{\cdot_Q}(F_{k_1})\otimes\Gamma_{\cdot_Q}(\overline{F}_{k_2})\big)(f_1\otimes f_2;\eta,\overline{\eta})\,,
	\end{split}
\end{equation}
where we used once more the linearity of $\Gamma_{\cdot_Q}^\mathbb{C}$ and $\Gamma_{\bullet_Q}^\mathbb{C}$. With reference to the perturbative coefficients as per Equation \eqref{Eq: Espressione coefficienti}, Equation \eqref{Eq: due punti} up to order $\mathcal{O}(\lambda^2)$ reads
\begin{equation*}
	\begin{split}
		\omega_{2}(\psi_\eta\llbracket\lambda\rrbracket;f_1\otimes f_2)&=\Gamma_{\bullet_Q}^\mathbb{C}(\Gamma_{\cdot_Q}^\mathbb{C}(\Phi)\otimes\Gamma_{\cdot_Q}^\mathbb{C}(\overline{\Phi}))(f_1\otimes f_2;\eta,\overline{\eta})\\
		&+\lambda\Gamma_{\bullet_Q}^\mathbb{C}(\Gamma_{\cdot_Q}^\mathbb{C}(G\circledast\overline{\Phi}\Phi^2)\otimes\Gamma_{\cdot_Q}^\mathbb{C}\overline{\Phi})(f_1\otimes f_2;\eta,\overline{\eta})\\
		&+\lambda\Gamma_{\bullet_Q}^\mathbb{C}(\Gamma_{\cdot_Q}^\mathbb{C}(\Phi)\otimes\Gamma_{\cdot_Q}^\mathbb{C}(\overline{G}\circledast\Phi\overline{\Phi}^2))(f_1\otimes f_2;\eta,\overline{\eta})\,.
	\end{split}
\end{equation*}
Taking into account that Equation \eqref{Eq: Gammadot on M_1} entails that $\Gamma_{\cdot_Q}^\mathbb{C}(\Phi)=\Phi$, $\Gamma_{\cdot_Q}^\mathbb{C}(\overline{\Phi})=\overline{\Phi}$,  Equation \eqref{eq_bullet2} yields that each separate term in the last expression reads
\begin{equation*}
	\Gamma_{\bullet_Q}^\mathbb{C}(\Phi\otimes\overline{\Phi})(f_1\otimes f_2;\eta,\overline{\eta})=(\Phi\bullet_Q\overline{\Phi})(f_1\otimes f_2;\eta,\overline{\eta})=(\Phi\otimes\overline{\Phi})(f_1\otimes f_2;\eta,\overline{\eta})+Q(f_1\otimes f_2)\,,
\end{equation*}
\begin{equation*}
	\begin{split}
		\Gamma_{\bullet_Q}^\mathbb{C}(\Gamma_{\cdot_Q}^\mathbb{C}(G\circledast\overline{\Phi}\Phi^2)&\otimes\overline{\Phi})(f_1\otimes f_2;\eta,\overline{\eta})=(G\circledast\overline{\Phi}\Phi^2+2G\circledast\overline{C}\,\Phi)\bullet_Q\overline{\Phi})(f_1\otimes f_2;\eta,\overline{\eta})\\
		&=(G\circledast(\overline{\Phi}\Phi^2+2\overline{C}\Phi)\otimes\overline{\Phi})(f_1\otimes f_2;\eta,\overline{\eta})+Q\cdot(G\circledast(2\overline{\Phi}\Phi)\otimes \mathbf{1})(f_1\otimes f_2;\eta,\overline{\eta})\\
		&+2\overline{Q}\cdot(G\circledast\overline{C}\mathbf{1}\otimes \mathbf{1})(f_1\otimes f_2;\eta,\overline{\eta})\,,
	\end{split}
\end{equation*}
\begin{equation*}
	\begin{split}
		\Gamma_{\bullet_Q}^\mathbb{C}(\Phi&\otimes\Gamma_{\cdot_Q}^\mathbb{C}(\overline{G}\circledast\Phi\overline{\Phi}^2))(f_1\otimes f_2;\eta,\overline{\eta})=(\Phi\bullet_Q(\overline{G}\circledast\Phi\overline{\Phi}^2+2\overline{G}\circledast C\,\overline{\Phi}))(f_1\otimes f_2;\eta,\overline{\eta})\\
		&=(\Phi\otimes \overline{G}\circledast(\Phi\overline{\Phi}^2+2C\overline{\Phi}))(f_1\otimes f_2;\eta,\overline{\eta})+\overline{Q}\cdot(\mathbf{1}\otimes \overline{G}\circledast(2\overline{\Phi}\Phi))(f_1\otimes f_2;\eta,\overline{\eta})\\
		&+2Q\cdot(\overline{G}\circledast C\mathbf{1}\otimes \mathbf{1})(f_1\otimes f_2;\eta,\overline{\eta})\,,
	\end{split}
\end{equation*}
where $\mathbf{1}\in\mathcal{D}^\prime(\mathbb{R}^{d+1};\mathsf{Pol}_{\mathbb{C}})$ is the identity functional defined in Example \eqref{Ex: Basic Functionals}. 
Gathering all contributions together and setting $\eta=0$, at first order in $\lambda$ the two-point correlation function of the perturbative solution reads
\begin{equation}
	\omega_2(\psi_0\llbracket\lambda\rrbracket;f_1\otimes f_2)=Q(f_1\otimes f_2)+2\lambda\overline{Q}(G\circledast\overline{C}\mathbf{1}\otimes\mathbf{1})(f_1\otimes f_2)+2\lambda Q\cdot(\overline{G}\circledast C\mathbf{1}\otimes \mathbf{1})(f_1\otimes f_2)+\mathcal{O}(\lambda^2)\,.
\end{equation}
\begin{remark}
	In this section we limited ourselves to constructing the two-point correlation function of the solution. 
	However it is important to highlight that, using the algebraic approach, we could have studied the $m$-point correlation function at any order in perturbation theory, for any $m\in\mathbb{N}$. 
	Yet, being the underlying random distribution Gaussian, all odd correlation functions vanish, while the even ones can be computed as
	\begin{equation*}
		\begin{split}
			&\omega_{m}(\psi\llbracket\lambda\rrbracket;f_1\otimes\ldots\otimes f_m)\\
			&=\bigl(\underbrace{\Psi_{\cdot_Q}\llbracket\lambda\rrbracket\bullet_{\Gamma_{\bullet_Q}^\mathbb{C}}\overline{\Psi}_{\cdot_Q}\llbracket\lambda\rrbracket\bullet_{\Gamma_{\bullet_Q}^\mathbb{C}}\ldots\bullet_{\Gamma_{\bullet_Q}^\mathbb{C}}\Psi_{\cdot_Q}\llbracket\lambda\rrbracket\bullet_{\Gamma_{\bullet_Q}^\mathbb{C}}\overline{\Psi}_{\cdot_Q}\llbracket\lambda\rrbracket}_{m\text{ pairs}}\bigr)(f_1\otimes\ldots\otimes f_m;0,0)\,,
		\end{split}
	\end{equation*}
	for all $f_1,\ldots, f_m\in\mathcal{D}(\mathbb{R}^{d+1})$. 
	The curly bracket encompasses $m$ pairs of terms of the form $\Psi_{\cdot_Q}\llbracket\lambda\rrbracket\bullet_{\Gamma_{\bullet_Q}^\mathbb{C}}\overline{\Psi}_{\cdot_Q}\llbracket\lambda\rrbracket$, with $m=2m'$ for any $m'\in\mathbb{N}$.
\end{remark}

\subsection{Renormalized equation}
As previously underlined, in order for the functional Equation \eqref{Eq: Functional equation} to account for the stochastic character codified by the complex white noise in Equation \eqref{Eq: Stochastic NLS}, we must deform the underlying algebra. In other words we act on Equation \eqref{Eq: Functional equation} with the deformation map $\Gamma_{\cdot_Q}^\mathbb{C}$, see Theorem \ref{Thm: deformation map cdot}, obtaining 
\begin{equation}\label{Eq: Equazione deformata}
	\begin{split}
		\Psi_{\cdot_Q}\llbracket\lambda\rrbracket&=\Gamma_{\cdot_Q}^\mathbb{C}(\Psi\llbracket\lambda\rrbracket)=\Gamma_{\cdot_Q}^\mathbb{C}(\Phi+\lambda G\circledast \overline{\Psi}\Psi^2)\\
		&=\Phi +\lambda G\circledast(\overline{\Psi}_{\cdot_Q}\cdot_Q\Psi_{\cdot_Q}\cdot_Q\Psi_{\cdot_Q})\,,
	\end{split}
\end{equation}
where $\Psi_{\cdot_Q}=\Gamma_{\cdot_Q}^\mathbb{C}(\Psi)$ and $\overline{\Psi}_{\cdot_Q}=\Gamma_{\cdot_Q}^\mathbb{C}(\overline{\Psi})$, while $\cdot_Q$ is defined in Equation \eqref{Eq: prodotto deformato2}. 
Yet, the presence of this deformed product makes quite cumbersome working with Equation \eqref{Eq: Equazione deformata}. Hence it is desirable to address whether one can rewrite such equation in an equivalent form in which only the pointwise product enters the game, namely we seek for the so called \textit{renormalized equation} \cite{Dappiaggi:2020gge}. The price to pay is the occurrence of new contributions to the equation known in the physics literature as \textit{counter-terms}. 
These are a by product both of the deformation procedure (and hence of the stochastic properties of the solution) and of the renormalization one.

We now present the fundamental result of this section, which is an extension of \cite[Prop. 6.6]{Dappiaggi:2020gge}.
\begin{theorem}\label{Thm: Equazione Rinormalizzata}
	Let us denote with $\Psi\llbracket\lambda\rrbracket\in\mathcal{A}^\mathbb{C}$ a perturbative solution of Equation \eqref{Eq: Functional equation} and with $\Psi_{\cdot_Q}=\Gamma_{\cdot_Q}^\mathbb{C}(\Psi)$ its counterpart lying in $\mathcal{A}^\mathbb{C}_{\cdot_Q}$. Then there exists a sequence of linear operator-valued functionals $\{M_k\}_{k\geq1}$ such that
	\begin{enumerate}
		\item \label{punto1} for all $k\geq1$ and for all $\eta,\eta^\prime\in\mathcal{E}(\mathbb{R}^{d+1})$, it holds: 
		\begin{equation}\label{eq_rinormalizzatac}
			M_k(\eta,\eta^\prime):\mathcal{E}(\mathbb{R}^{d+1})\times\mathcal{E}(\mathbb{R}^{d+1})\rightarrow\mathcal{E}(\mathbb{R}^{d+1})\times\mathcal{E}(\mathbb{R}^{d+1}).
		\end{equation}
		\item \label{punto2} Every element of the sequence has a polynomial dependence on $\eta,\eta^\prime$ and, for all $i,j,k\in\mathbb{N}$ such that $i+j=2k+1$, $M_k^{(i,j)}(0)=0$.
	\end{enumerate}
In addition, calling $M:=\sum_{k\geq1}\lambda^kM_k$, $\Psi_{\cdot_Q}$ is a solution of
		\begin{equation}\label{Eq: Equazione rinormalizzata}
			\Psi_{\cdot_Q}=\Phi+\lambda G\circledast\overline{\Psi}_{\cdot_Q}\Psi_{\cdot_Q}^2+G\circledast (M\Psi_{\cdot_Q}).
		\end{equation}
\end{theorem}

\begin{proof} Referring to the formal power series expansion
\begin{equation}
	M:=\sum_{k\geq1}\lambda^kM_k,
\end{equation}
we proceed inductively with respect to the perturbative order $k$,
providing at every inductive step a suitable candidate $M_k$ satisfying the desired properties. On account of Equation \eqref{Eq: Gammadot on M_1}, at order $k=0$ the solution reads $\Psi_{\cdot_Q}=\Phi+\mathcal{O}(\lambda)$ and no correction is required. We can use the expression of $F_1$ in Equation \eqref{Eq: Espressione coefficienti} to compute $M$ at first order, namely
\begin{equation*}
	\begin{split}
		\mathrm{LHS}&=\Gamma_{\cdot_Q}^\mathbb{C}(\Phi+\lambda G\circledast\overline{\Phi}{\Phi}^2)+\mathcal{O}(\lambda^2)\\
		&=\Phi+\lambda G\circledast\overline{\Phi}\Phi^2+2\lambda G\circledast\overline{C}\Phi+\mathcal{O}(\lambda^2)\,,
	\end{split}
\end{equation*}
\begin{equation*}
	\begin{split}
		\mathrm{RHS}&=\Phi+\lambda G\circledast\overline{\Phi}\Phi^2+\lambda G\circledast M_1\Phi+\mathcal{O}(\lambda^2)\,,
	\end{split}
\end{equation*}
where $\overline{C}\in\mathcal{E}(\mathbb{R}^{d+1})$ has been defined in the proof of Theorem \ref{Thm: deformation map cdot} and where with $\mathrm{LHS}$ and $\mathrm{RHS}$ we denote the left hand side and the right hand side, respectively, of Equation \eqref{Eq: Equazione rinormalizzata}.
Comparing these expressions, we conclude $M_1=2\overline{C}\mathbb{I}$ where $\mathbb{I}$ is the identity operator. 
All the properties listed in the statement of the theorem, as well as Equation \eqref{Eq: Equazione rinormalizzata} up to order $\mathcal{O}(\lambda^2)$ are fulfilled as one can infer by direct inspection. 
To complete the induction procedure we assume that $M_{k'}$ has been defined for all $k'\leq k-1$, $k\in\mathbb{N}$ so to satisfy all properties listed in the statement and so that Equation \eqref{Eq: Equazione rinormalizzata} holds true up to order $\mathcal{O}(\lambda^{k})$. 
We show how one can individuate $M_k$ sharing the same properties up to order $\mathcal{O}(\lambda^{k+1})$. 
Mimicking the procedure adopted for $k=1$ it descends
\begin{equation*}
	\begin{split}
		\mathrm{LHS}=T_{k-1}&+\lambda^k\sum_{k_1+k_2+k_3=k-1}G\circledast\Gamma_{\cdot_Q}^\mathbb{C}(\overline{F}_{k_1}F_{k_2}F_{k_3})+\mathcal{O}(\lambda^{k+1})\,,
	\end{split}
\end{equation*}
\begin{equation*}
	\begin{split}
		\mathrm{RHS}=T_{k-1}&+\lambda^k\sum_{k_1+k_2+k_3=k-1}G\circledast\bigl[\Gamma_{\cdot_Q}^\mathbb{C}(\overline{F}_{k_1})\Gamma_{\cdot_Q}(F_{k_2})\Gamma_{\cdot_Q}^\mathbb{C}(F_{k_3})\bigr]\\
		&+\lambda^k\sum\limits_{\substack{k_1+k_2=k\\k_2\neq0}}G\circledast(M_{k_1}\Gamma_{\cdot_Q}^\mathbb{C}(F_{k_2}))+\lambda^kG\circledast(M_k\Phi)+\mathcal{O}(\lambda^{k+1})\,.
	\end{split}
\end{equation*}
Here $T_{k-1}$ summarizes all contributions up to order $\mathcal{O}(\lambda^k)$ and thanks to the inductive hypothesis it satisfies Equation \eqref{Eq: Equazione rinormalizzata}. Thus it does not contribute to the construction of $M_k$.

In the proof of Theorem \ref{Thm: valor medio nullo} we introduced the vector space $\mathcal{O}\subset\mathcal{A}^\mathbb{C}$ of distributions $u\in\mathcal{A}^\mathbb{C}$ having odd polynomial degree in $\Phi,\overline{\Phi}\in\mathcal{A}^\mathbb{C}$. 
Since the inductive hypothesis entails that $M_l$ has an even polynomial degree in both $\Phi$ and $\overline{\Phi}$ for all $l\leq k-1$, it follows that, since $\Gamma_{\cdot_Q}^\mathbb{C}(F_{k_2})\in\mathcal{O}$, then $M_{k_1}\Gamma_{\cdot_Q}^\mathbb{C}(F_{k_2})\in\mathcal{O}$.

We can further characterize the perturbative coefficients as particular elements of the vector space $\mathcal{O}$: The specific form of the nonlinear potential $V(\psi,\overline{\psi})=\overline{\psi}\psi^2$ implies that every $F_l$, $l\in\mathbb{N}$ depends precisely on $m$ fields $\overline{\Phi}$ and on $m+1$ fields $\Phi$, with $m\in\mathbb{N}$ such that $m\leq l$. This characterization entails that there is always an unpaired field $\Phi$. Note also that this property is preserved by the action of $\Gamma_{\cdot_Q}^\mathbb{C}$ since contractions collapse pairs of $\Phi$ and $\overline{\Phi}$. We conclude that, besides the term involving $M_k$, all  other contributions in both sides of Equation \eqref{Eq: Equazione rinormalizzata} are of the form $P\circledast u$, $u\in\mathcal{O}$. Moreover, $u$ can be written as $u=K\Phi$ with $K$ a linear operator abiding to all the properties listed in the statement of the theorem for the renormalization counterterms. Thus we can write
\begin{equation*}
	\begin{split}
		\Psi_{\cdot_Q}-\Phi-&\lambda G\circledast\Psi_{\cdot_Q}+\lambda G\circledast\overline{\Psi}_{\cdot_Q}\Psi_{\cdot_Q}^2-G\circledast (M\Psi_{\cdot_Q})\\
		&=\lambda^kG\circledast\bigl[K-M_k\bigr]\Phi+\mathcal{O}(\lambda^{k+1}).
	\end{split}
\end{equation*}
If we set $M_k=K$, it suffices to observe that Equation \eqref{Eq: Equazione rinormalizzata} holds true up to order $\mathcal{O}(\lambda^{k})$. This concludes the induction procedure and the proof.
\end{proof}

\begin{remark}
We refer to \cite[Rem. 6.7]{Dappiaggi:2020gge} for the explicit computation of the coefficient $M_2$ in the real scenario. 
The complex case can be discussed similarly.
\end{remark}

\begin{remark}
	One may wonder why Equation \eqref{Eq: Equazione rinormalizzata} does not encompass a correction involving $\overline{\Psi}_{\cdot_Q}$, namely a term of the form $G\circledast(N\overline{\Psi}_{\cdot_Q})$  with $N$  the same defining properties of $M$. Although, a priori such kind of contribution should be taken into account, Theorem \ref{Thm: Equazione Rinormalizzata}  entails that such additional term is not necessary. Furthermore, our renormalization procedure makes a direct contact with analogous results within other non-perturbative frameworks, such as the theory of regularity structures and of paracontrolled calculus. More precisely, in full analogy with \cite{Hoshino}, only correction terms linear in $\Psi$ should be expected.
\end{remark}

\subsection{A graphical approach to renormalizability}\label{Sec: graph}

As we have seen in the previous sections, the construction of the algebra $\mathcal{A}^\mathbb{C}_{\cdot_Q}$, which is a necessary step in order to compute at any order in perturbation theory the expectation value of the perturbative solution, requires renormalization in order to be meaningful. 

As a matter of fact, in the construction of the perturbative solution only some elements of the algebra $\mathcal{A}^\mathbb{C}_{\cdot_Q}$ are involved.
The question we address in this section is whether we can find a condition under which the number of ill-defined contributions to be renormalized in order to compute the expectation value of the solution remains finite to all orders in $\lambda$.

As a consequence, such a condition would tell whether one needs a finite or infinite number of independent renormalization steps in order to compute the expectation value of the solution at any order in perturbation theory.

This is tantamount to classifying a model as \emph{super-renormalizable} or, equivalently, \emph{sub-critical}.

It is useful to tackle such problem using a graphical representation, taking advantage of an analogous formulation for the stochastic $\varphi^3_d$ model analyzed in \cite{Dappiaggi:2020gge}. We can associate to every perturbative contribution a graph via the following prescription:
\begin{itemize}
	\item the occurrence of the field $\Phi\in\mathcal{A}^\mathbb{C}$ is represented by the symbol $\fiammifero$, while the complex conjugate field $\overline{\Phi}$ corresponds to $\fiammiferoCC$.
	\item The convolution with $G$ and $\overline{G}$ is denoted by segments $\propagatore$ and $\propagatoreCC$ , respectively.
	\item  The graphical analogue of the pointwise product between elements of $\mathcal{A}^\mathbb{C}$ consists of joining the roots of the respective graphs into a single vertex.
\end{itemize}
\begin{example}
	In order to better grasp the notation, consider the following two examples:
	\begin{equation*}
		\Phi\overline{\Phi}=
		\begin{tikzpicture}[thick,scale=1.2]
			\draw[red] (0,0) -- (0.2,0.3);
			\filldraw[red] (0.2,0.3)circle (1pt);
			\draw (0,0) -- (-0.2,0.3);
			\filldraw (-0.2,0.3)circle (1pt);
		\end{tikzpicture},
		\hspace{3cm}
		G\circledast\Phi\overline{\Phi}=
		\begin{tikzpicture}[thick,scale=1.2]
			\draw[red] (0,0) -- (0.2,0.3);
			\filldraw[red] (0.2,0.3)circle (1pt);
			\draw (0,0) -- (-0.2,0.3);
			\filldraw (-0.2,0.3)circle (1pt);
			\draw (0,0) -- (0,-0.3);
		\end{tikzpicture}.
	\end{equation*}
\end{example}

The rules outlined above encompass all the algebraic structures required to represent the perturbative coefficients, but the need of renormalization occurs when acting with the deformation map $\Gamma_{\cdot_Q}^\mathbb{C}$.  Its evaluation against elements lying in $\mathcal{A}^\mathbb{C}$ has been thoroughly studied in Section \ref{Sec: expectations}. The net effect amounts to a sequence of contractions between $\Phi$ and $\overline{\Phi}$, related to the occurrence of powers of $Q$ and $\overline{Q}$. Hence at the graphical level it translates into progressively collapsing pairs of leaves into a single loop. 

Observe how, in view of the stochastic properties of the complex white noise, see Equation \eqref{Eq: complex white noise}, only contractions between leaves of different colours are allowed.
\begin{example}
	The graphical counterpart of $\Gamma_{\cdot_Q}(\Phi\overline{\Phi})(f;\eta,\overline{\eta})=\Phi\overline{\Phi}(f;\eta,\overline{\eta})+C(f)$, $f\in\mathcal{D}(\mathbb{R}^{d+1}),\eta\in\mathcal{E}(\mathbb{R}^{d+1})$, is
	\begin{equation*}
		\Gamma_{\cdot_Q}(\Phi\overline{\Phi})=
		\begin{tikzpicture}[thick,scale=1.2]
			\draw[red] (0,0) -- (0.2,0.3);
			\filldraw[red] (0.2,0.3)circle (1pt);
			\draw (0,0) -- (-0.2,0.3);
			\filldraw (-0.2,0.3)circle (1pt);
		\end{tikzpicture}
		+\fish.
	\end{equation*}
\end{example}

The next step consists of identifying all possible graphs appearing in the perturbative decomposition of a solution. From a direct inspection of the  coefficients $F_k$, $k\in\mathbb{N}$ one can prove inductively that all such graphs have a tree structure with branches of the form
\begin{equation*}
	\begin{tikzpicture}[thick,scale=1.2]
		\draw (0,0) -- (0,0.3);
		\draw[red] (0,0.3) -- (0.3,0.6);
		\filldraw[red] (0.3,0.6)circle (1pt);
		\draw (0,0.3) -- (0,0.7);
		\filldraw (0,0.7)circle (1pt);
		\draw (0,0.3) -- (-0.3,0.6);
		\filldraw (-0.3,0.6)circle (1pt);
	\end{tikzpicture}
	\hspace{1cm}\text{or}\hspace{1cm}
	\begin{tikzpicture}[thick,scale=1.2]
		\draw[red] (0,0) -- (0,0.3);
		\draw[red] (0,0.3) -- (0.3,0.6);
		\filldraw[red] (0.3,0.6)circle (1pt);
		\draw[red] (0,0.3) -- (0,0.7);
		\filldraw[red] (0,0.7)circle (1pt);
		\draw (0,0.3) -- (-0.3,0.6);
		\filldraw (-0.3,0.6)circle (1pt);
	\end{tikzpicture}.
\end{equation*}
Furthermore, increasing by one the perturbative order in $\lambda$ translates into adding a vertex with this topology. Such result entails a constraint on the class of graphs involved in the construction of a solution and it allows us to perform an analysis with respect to the perturbative order. This differs from the approach adopted in \cite{Dappiaggi:2020gge}, where the authors consider a larger class, including also graphs which do not enter the construction. As a results we improve the bound on the renormalizability condition.

\begin{remark}
	Non contracted leaves correspond, at the level of distributions, to the multiplication by a smooth function $\eta\in\mathcal{E}(\mathbb{R}^{d+1})$, which does not contribute to the divergence of the graph. Hence we restrict ourselves to maximally contracted graphs, being aware that adding uncontracted leaves does not affect the degree of divergence.
\end{remark}

\begin{definition}
	A graph is said to be admissible if it derives from a graph associated to any perturbative coefficient $F_k$ via a maximal contraction of its leaves.
\end{definition}

Following the prescription presented above, we construct a graph $\mathcal{G}$ and an associated distribution $u_{\mathcal{G}}$ completely characterized by the following features:
\begin{itemize}
	\item $\mathcal{G}$ has $L$ edges and $N$ vertices of valence $1$ or $4$. Here the valence of a vertex is the number of edges connected to it.
	\item Each edge $e$ of the graph is the pictorial representation of a propagator $G(x_{s(e)},x_{t(e)})$, where $s(e)$ stands for the origin of $e$, while with $t(e)$ we denote the target,
	\item  the integral kernel of $u_{\mathcal{G}}$ is a product of propagators:
	\begin{equation}\label{eq_propagatori}
		u_{\mathcal{G}}(x_1,\dots,x_M):=\prod_{e\in E_{\mathcal{G}}}G(x_{s(e)},x_{t(e)}),
	\end{equation}
	where $E_{\mathcal{G}}=\{e_i, i=1,...,N\}$ represents the set of edges of $\mathcal{G}$.
\end{itemize}

When acting with $\Gamma_{\cdot_Q}^\mathbb{C}$ on elements of the perturbative expansion, we are working with distributions $u\in\mathcal{D}^\prime(U)$, $U\subseteq\mathbb{R}^{N(d+1)}$. As outlined in the proof of Theorem \ref{Thm: deformation map cdot}, the singular support of $u$ corresponds to \clacomment{$\Lambda_t^{N}=\{(t_1,x_1,\ldots,t_N,x_N)\in(\mathbb{R}^{d+1})^N\,\vert\,t_1=\ldots=t_N\}$} and renormalization can be dealt with via microlocal analytical techniques based on the calculation of the scaling degree of $u$ with respect to the submanifold $\Lambda_t^N$ of $\mathbb{R}^{N(d+1)}$. A direct computation yields\\
\clacomment{
\begin{equation}\label{Eq: Degree of divergence G}
	\rho_{\text{Diag}_{N(d+1)}}(u_\mathcal{G}):=Ld-2(N-1)\,,
\end{equation}
where $\rho_{\Lambda_t^N}(u_\mathcal{G})$ is the weighted degree of divergence of $u_{\mathcal{G}}$ with respect to $\Lambda_t^N$, while $2(N-1)$ is the weighted co-dimension of $\Lambda_t^N$, see \cite[App. B]{Dappiaggi:2020gge}.}

Our goal is to find an expression of Equation \eqref{Eq: Degree of divergence G} in terms of the perturbative order $k$. We divide the analysis in several steps.\\

\noindent \textbf{Step 1: L(k)}.\hspace{1mm} Since the number of edges is not affected by contractions into a single vertex, at this level we can ignore the action of $\Gamma_{\cdot_Q}$. We state that increasing by one the perturbative order corresponds to adding three edges arranged in one of the following admissible branches:
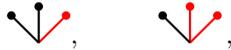
\begin{equation*}
	\begin{tikzpicture}[thick,scale=1.2]
		\draw[red] (0,0) -- (0.3,0.3);
		\filldraw[red] (0.3,0.3)circle (1pt);
		\draw (0,0) -- (0,0.4);
		\filldraw (0,0.4)circle (1pt);
		\draw (0,0) -- (-0.3,0.3);
		\filldraw (-0.3,0.3)circle (1pt);
	\end{tikzpicture},
	\hspace{1cm}
	\begin{tikzpicture}[thick,scale=1.2]
		\draw[red] (0,0) -- (0.3,0.3);
		\filldraw[red] (0.3,0.3)circle (1pt);
		\draw[red] (0,0) -- (0,0.4);
		\filldraw[red] (0,0.4)circle (1pt);
		\draw (0,0) -- (-0.3,0.3);
		\filldraw (-0.3,0.3)circle (1pt);
	\end{tikzpicture}\,,
\end{equation*}
that is, the number of vertices of valence $4$ coincides with the perturbative order. This implies that
\begin{equation}\label{Eq: L(k)}
	L(k)=3k+1.
\end{equation}

\noindent \textbf{Step 2: N(k)}.\hspace{0.1cm} The value $N(k)$ does not depend on how the branches are arranged. It is manifest that switching a branch from a vertex to another leaves $N(k)$ unchanged. For this reason in what follows we refer to tree-like graphs where all vertices of a layer are saturated before moving to the next one. As an example, we do not consider a graph with the following shape:
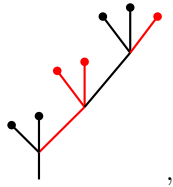
\begin{equation*}
	\begin{tikzpicture}[thick,scale=1.2]
		\draw (0,0) -- (0,0.3);
		\draw[red] (0,0.3) -- (0.5,0.8);
		\draw[red] (0.5,0.8) -- (0.2,1.2);
		\filldraw[red] (0.2,1.2)circle (1pt);
		\draw[red] (0.5,0.8) -- (0.5,1.3);
		\filldraw[red] (0.5,1.3)circle (1pt);
		\draw (0.5,0.8) -- (1,1.4);
		
		\draw (1,1.4) -- (0.7,1.8);
		\filldraw (0.7,1.8)circle (1pt);
		\draw (1,1.4) -- (1,1.9);
		\filldraw (1,1.9)circle (1pt);
		\draw[red] (1,1.4) -- (1.3,1.8);
		\filldraw[red] (1.3,1.8)circle (1pt);
		\draw (0,0.3) -- (0,0.7);
		\filldraw (0,0.7)circle (1pt);
		\draw (0,0.3) -- (-0.3,0.6);
		\filldraw (-0.3,0.6)circle (1pt);
	\end{tikzpicture}\,,
\end{equation*}
but rather the simpler
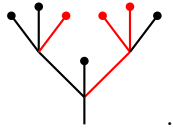
\begin{equation*}
	\begin{tikzpicture}[thick,scale=1.2]
		\draw (0,0) -- (0,0.3);
		\draw[red] (0,0.3) -- (0.5,0.8);
		\draw[red] (0.5,0.8) -- (0.2,1.2);
		\filldraw[red] (0.2,1.2)circle (1pt);
		\draw[red] (0.5,0.8) -- (0.5,1.3);
		\filldraw[red] (0.5,1.3)circle (1pt);
		\draw (0.5,0.8) -- (0.8,1.2);
		\filldraw (0.8,1.2)circle (1pt);
		\draw (0,0.3) -- (0,0.7);
		\filldraw (0,0.7)circle (1pt);
		
		\draw (0,0.3) -- (-0.5,0.8);
		\draw (-0.5,0.8) -- (-0.8,1.2);
		\filldraw (-0.8,1.2)circle (1pt);
		\draw (-0.5,0.8) -- (-0.5,1.3);
		\filldraw (-0.5,1.3)circle (1pt);
		\draw[red] (-0.5,0.8) -- (-0.2,1.2);
		\filldraw[red] (-0.2,1.2)circle (1pt);
	\end{tikzpicture}\,.
\end{equation*}
Thanks to an argument similar to that of the derivation of $L(k)$ and reasoning on these simplified configurations, one can infer that
\begin{equation*}
	N(k)=3k+1.
\end{equation*}
Yet, this holds true for non-contracted graphs, while we are looking for configurations showing the highest number of divergences. Since the number of vertices of valence 1 at fixed perturbative order $k$ is $2k+1$ and, in view of the observation that there remains always a single non-contracted leaf, performing all permissible contractions amounts to removing $\frac{1}{2}(2k+1-1)=k$ vertices, namely
\begin{equation}\label{Eq: N(k)}
	N(k)=3k+1-k=2k+1.
\end{equation}
\noindent \textbf{Step 3}.\hspace{0.1cm} We can express the perturbative order $k$ in terms of $N$ by inverting Equation \eqref{Eq: N(k)}, namely
\begin{equation}
	k=\frac{N-1}{2}\,.
\end{equation}
Since every admissible configuration has an odd number of vertices, $k(N)$ is an integer. Inserting this expression of the perturbative order in Equation \eqref{Eq: L(k)} we obtain
\begin{equation}\label{L(N)}
	L(N)=\frac{3}{2}N-\frac{1}{2}\,.
\end{equation}
We are interested in the scenarios when only a finite number of graphs shows a singular behaviour. In the present setting this condition translates into requiring that, for $N$ large enough, the degree of divergence of admissible graphs becomes negative.
\begin{theorem}
	If \clacomment{$d=1$}, only a finite number of graphs in the perturbative solution of Equation \eqref{Eq: Functional equation} needs to be renormalized.
\end{theorem}
\begin{proof}
	\clacomment{
	In view of Equation \eqref{Eq: Degree of divergence G} and of the relations derived in the preceding steps, it descends
	\begin{equation}\label{Eq: rho}
		\begin{split}
			\rho(u_{\mathcal{G}})&=Ld-2(N-1)=\Bigl(\frac{3}{2}N-\frac{1}{2}\Bigr)d-2(N-1)\\
			&=N\Bigl(\frac{3}{2}d-2\Bigr)-\frac{d}{2}+2\,.
		\end{split}
	\end{equation}
	For the condition $\rho(u_{\mathcal{G}})<0$ to occur, the coefficient in front of $N$ must be negative, so that for sufficiently large $N$ the inequality holds true.} \clacomment{This translates to the inequality
	\begin{equation}\label{Eq: Subcritical regime}
		\frac{3}{2}d-2<0\;\Rightarrow\;d<\frac{4}{3}\,.
	\end{equation}
}
\end{proof}
\begin{remark}
	Once the spatial dimension $d$ has been fixed, Equation \eqref{Eq: rho} allows us to find the minimal number of vertices needing renormalization and for which further divergences do not occur. 
	Subsequently, via Equation \eqref{Eq: N(k)}, we can also infer the maximum number of graphs requiring renormalization.
\end{remark}

\begin{remark}
Even if in our analysis we focused on the stochastic Schr\"odinger equation with a cubic nonlinearity, this approach to the study of the subcritical regime can be extended to a more general polynomial potential of the form $\vert\psi\vert^{2\kappa}\psi$ with $\kappa>1$, as in Equation \eqref{Eq: Stochastic NLS}. At a graphical level the net effect is that admissible graphs present vertices of valence at most $2\kappa+2$. Hence the derivation of Equation \eqref{Eq: Subcritical regime} follows slavishly the preceding steps, yielding that for $d=1$ we are still considering a subcritical regime, no matter the value of $\kappa$.
\end{remark}

\paragraph{Acknowledgements.}
We are thankful to N. Drago and V. Moretti for helpful discussions. A.B. is supported by a PhD fellowship of the University of Pavia, while P.R by a postdoc fellowship of the Institute for Applied Mathematics of the University of Bonn.

\end{document}